\documentclass[conference]{IEEEtran}
\IEEEoverridecommandlockouts
\usepackage{cite}
\usepackage{amsmath,amsthm,amssymb,amsfonts}
\usepackage{algorithmic}
\usepackage{graphicx}
\usepackage{textcomp}
\usepackage{xcolor}
\usepackage[hyphens]{url}

\usepackage{caption}
\usepackage{algorithm}
\usepackage{makecell}
\usepackage{subcaption}
\usepackage{booktabs}
\usepackage{multirow}
\usepackage{siunitx}
\usepackage{dblfloatfix}
\usepackage{flushend}
\renewcommand{\figurename}{Figure}

\newtheorem{proposition}{Proposition}
\def\BibTeX{{\rm B\kern-.05em{\sc i\kern-.025em b}\kern-.08em
    T\kern-.1667em\lower.7ex\hbox{E}\kern-.125emX}}
\begin{document}

\title{\LARGE{OASIS: Offsetting Active Reconstruction Attacks in Federated Learning
}\\
}



\author{\IEEEauthorblockN{Tre' R. Jeter\IEEEauthorrefmark{2}\IEEEauthorrefmark{1}, Truc Nguyen\IEEEauthorrefmark{3}\IEEEauthorrefmark{1}, Raed Alharbi\IEEEauthorrefmark{4}, and My T. Thai\IEEEauthorrefmark{2}}

\IEEEauthorblockA{
\IEEEauthorrefmark{2}University of Florida, Gainesville, FL 32611, USA\\
\IEEEauthorrefmark{3}National Renewable Energy Laboratory, Golden, CO 80401, USA\\
\IEEEauthorrefmark{4}Saudi Electronic University, Riyadh, Saudi Arabia\\
Email: t.jeter@ufl.edu, Truc.Nguyen@nrel.gov, ri.alharbi@seu.edu.sa, mythai@cise.ufl.edu}

\thanks{\IEEEauthorrefmark{1}
These authors contributed equally to this work.
}
}


\maketitle

\begin{abstract}
    Federated Learning (FL) has garnered significant attention for its potential to protect user privacy while enhancing model training efficiency. For that reason, FL has found its use in various domains, from healthcare to industrial engineering, especially where data cannot be easily exchanged due to sensitive information or privacy laws. However, recent research has demonstrated that FL protocols can be easily compromised by active reconstruction attacks executed by dishonest servers.  These attacks involve the malicious modification of global model parameters, allowing the server to obtain a verbatim copy of users' private data by inverting their gradient updates. Tackling this class of attack remains a crucial challenge due to the strong threat model. In this paper, we propose a defense mechanism, namely OASIS, based on image augmentation that effectively counteracts active reconstruction attacks while preserving model performance. We first uncover the core principle of gradient inversion that enables these attacks and theoretically identify the main conditions by which the defense can be robust regardless of the attack strategies. We then construct our defense with image augmentation showing that it can undermine the attack principle. Comprehensive evaluations demonstrate the efficacy of the defense mechanism highlighting its feasibility as a solution.
\end{abstract}


\section{Introduction}
In recent years, Federated Learning (FL) has developed into a well-respected distributed learning framework that promotes user privacy 
with high model performance. By design, FL authorizes collaborative training of a global model between millions of users without revealing any of their locally trained, private data. It is an iterative protocol where, in each round, a central server distributes the most up-to-date global model to an arbitrary subset of users that train locally and communicate their model updates back to the server. These model updates include the gradients that are calculated based on the global model and the local training data. The central server then averages these model updates to form a new global model to distribute in the next round. 

\begin{figure*}[htp!]
    \centering
    \includegraphics[width=0.9\linewidth]{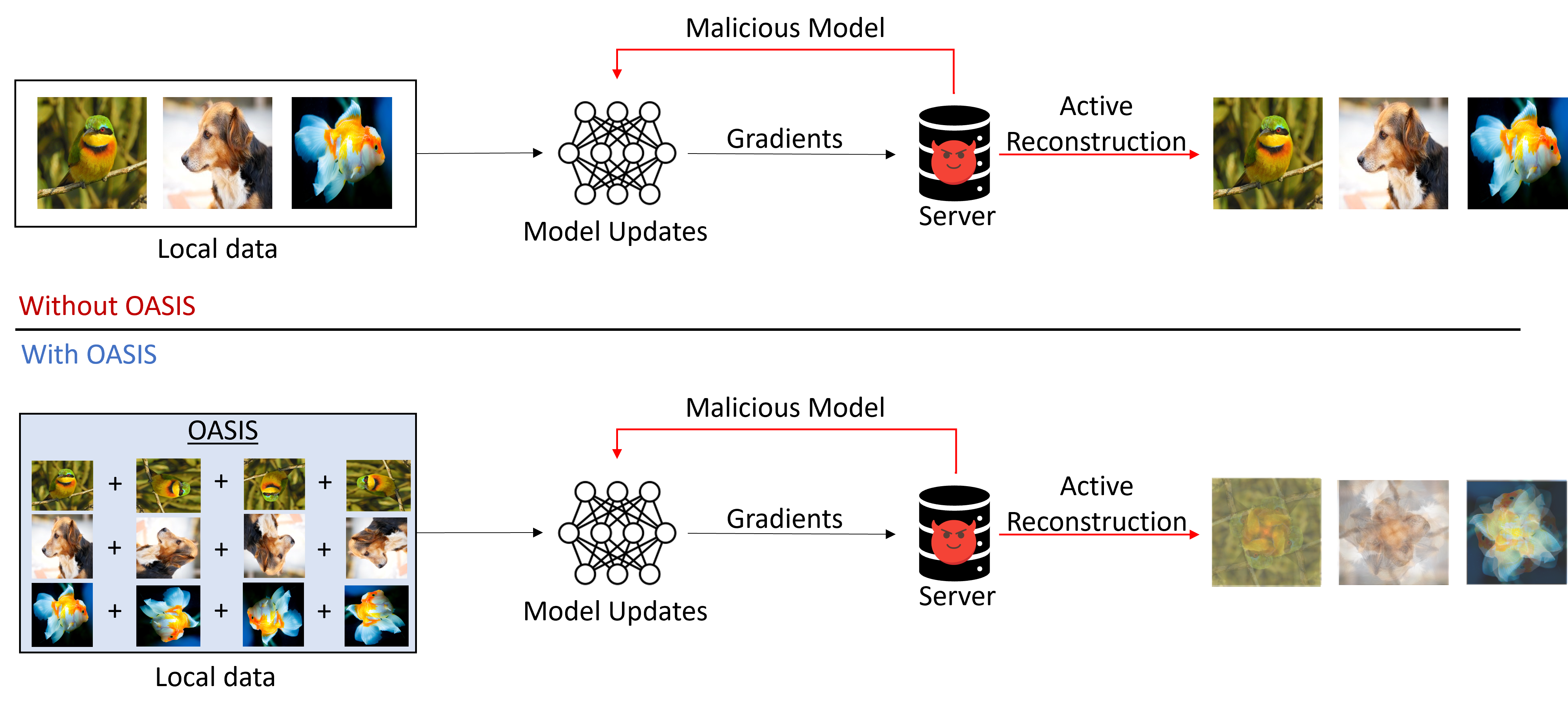}
    \caption{Overview design of OASIS. 
    \textit{Top:} Standard \textit{active} reconstruction attack with malicious model modifications 
    perfectly reconstructing 
    training samples. \textit{Bottom:} OASIS in place with augmented data to defend the \textit{active} reconstruction attacks. The resulting reconstruction is a linear combination of images, effectively hiding the content of training samples. \emph{Note: Rotation is not the only transformation within OASIS.}
    }
    \label{fig:system-design}
\end{figure*}

With a disruptive privacy-centric design, FL has been regarded as an auspicious solution for applying machine learning to the healthcare sector, particularly in scenarios where sharing medical data between different sites is intractable due to strict privacy protection policies such as the Health Insurance Portability and Accountability Act (HIPAA) \cite{moore2019review} and General Data Protection Regulation (GDPR) \cite{voigt2017eu}. Numerous studies have proposed FL for medical image analysis, utilizing data such as X-rays, MRIs, and PET scans from different hospital sites while complying with privacy laws \cite{kaissis2020secure,adnan2022federated,li2020multi,silva2019federated,jiang2022harmofl,stripelis2021scaling}. This innovative approach is not limited to healthcare; FL is also making significant strides in industrial engineering. For instance, in urban environment image sensing, research has shown that FL makes it easier to perform a time-series analysis of industrial environment factors obtained from multiple sensors and unmanned aerial vehicles (UAVs) across different companies while maintaining confidential data privacy \cite{hu2018federated,liu2020federated,tam2021adaptive,gao2020federated}. Beyond these applications, FL is also stimulating advancements in diverse domains such as control systems, autonomous vehicles, and smart manufacturing \cite{
huong2021detecting,zeng2022federated,kevin2021federated}, showcasing the versatility and broad impact of a privacy-preserving learning framework.

However, the promise of privacy for clients in FL has been constantly challenged \cite{nguyen2023preserving}. Recent work \cite{boenisch2023curious,fowl2021robbing,nguyen2022blockchain,pasquini2022eluding,tre2023privacy,nguyen2023active,vu2024analysis,vu2023active} has investigated a strong and practical threat model in which the server can be \textit{actively dishonest}, such that it is capable of maliciously modifying the global model before dispatching it to the users. This threat model has instigated several \textit{active} reconstruction attacks in which an FL server can perfectly reconstruct some data points in a users' training data \cite{boenisch2023curious,fowl2021robbing,vu2023active}. These attacks exploit a fundamental concept that the gradients in the local model updates sent by users may contain complete and memorized individual data points. These gradients can later be inverted by the server to reveal such data points. 
As an actively dishonest adversary, the server can strategically manipulate the weights of the global model to maximize the number of individual data points that can be reconstructed from the users' gradients. For that reason, although the training data is said to never leave a users' device, it can still be reconstructed, 
thereby refuting the claim of privacy-preservation in FL.

Given such a strong adversary, defending against this class of \textit{active} reconstruction attacks is challenging. 
Until now, a mitigation approach for FL-based attacks focused on obfuscating the gradients via a Differential Privacy (DP) mechanism, such as DPSGD \cite{abadi2016deep}, that formally bounds the privacy leakage by adding calibrated noise to the gradients. 
However, previous work \cite{fowl2021robbing,boenisch2023curious} 
has shown that 
to prevent an attacker from discerning the content of reconstructed data, the user must add a significant amount of DP noise to the gradients that unfortunately 
degrades the overall model performance. 

In this paper, we propose a new defense, OASIS, to \textbf{\underline{O}}ffset this class of \textbf{\underline{A}}ctive recon\textbf{\underline{S}}truct\textbf{\underline{I}}on attack\textbf{\underline{S}}. As there are different strategies to manipulate the global model for conducting this attack, it is imperative to figure out how to tackle the attack in principle so that the defense is robust regardless of the manipulation strategies. We first analyze the attack surface and determine the core vulnerability in the gradient updates that enables the memorization of individual training samples. By doing so, we generalize the existing attacks by discovering the conditions under which a dishonest server can conduct \textit{gradient inversion} to reconstruct users' data. We then intuitively show how to undermine those conditions and mitigate the impact of the attacks.

From the attack principle, we show that the users can preprocess their training data in a way that prevents the samples from being revealed via gradient inversion, effectively countering this class of \textit{active} reconstruction attacks. A mechanism for such preprocessing 
is 
image augmentation 
\cite{ho2019population,yarats2020image}. This includes adding 
augmented versions of an image, such as rotated, flipped, and sheared counterparts to the training data before computing the gradients. By doing this, OASIS aims to have the gradients memorize a \textit{linear combination} of the original image and its augmented versions, instead of memorizing any individual images. As a result, inverting these gradients would reconstruct 
what 
appears to be an overlap of multiple images, 
thereby effectively preventing the server from discerning the content of the reconstructed images, as shown in Figure \ref{fig:system-design}. Since image augmentation is used to improve model generalization \cite{cubuk2019autoaugment}, we 
safely maintain the performance of FL with this countermeasure. Our analysis shows that OASIS opens a new approach to protect users' data from gradient inversion without suffering the utility loss as in DP.

\noindent \textbf{Contributions.} Our key contributions are as follows:
\begin{itemize}
    \item We analyze the attack surface and determine the key principle behind gradient inversion that enables 
    \textit{active} reconstruction attacks. From that, we theoretically show how to 
    tackle this class of attack, regardless of how the attacker manipulates the global model parameters.
    \item Based on the attack principle, we present OASIS as a suite of image augmentations. To our knowledge, this mechanism stands as the first general and scalable defense against \textit{active} reconstruction attacks via gradient inversion by \textit{actively dishonest servers} in FL.
    \item We thoroughly analyze the effectiveness of OASIS through experiments with respect to attack success rate and augmentation type. 
    We also show how OASIS maintains model performance.
\end{itemize}

\noindent \textbf{Organization.} The paper's structure is as follows: Section \ref{sec:prelim} provides a primer on FL, image augmentation, and the main augmentations used. Section \ref{sec:defense} presents the threat model, attack principle, and our OASIS defense. Section \ref{sec:exp} presents an in-depth experimental analysis and results supporting our defense. Section \ref{sec:related} discusses related research on reconstruction attacks and existing defenses. 
Section \ref{sec:con} concludes the paper, summarizing our key findings.

\section{Preliminaries}\label{sec:prelim}
In this section, we summarize the
FL process while also describing the benefits of image augmentation during training.

\subsection{\textbf{Federated Learning}}
Depending on how training data is distributed among the participants, there are two main versions of FL: horizontal and vertical. In this paper, we focus on a horizontal setting in which different data owners hold the same set of features but different sets of samples. 
We denote $f_w : \mathbb{R}^d \rightarrow \mathbb{R}^k$ as a $k$-class neural network model that is parameterized by a set of weights $w$. The goal of $f_w$ is to map a data point $x \in \mathbb{R}^d$ to a vector of posterior probabilities $f_w(x_i) = \mathcal{Y}$ over $k$ classes.

FL is an iterative learning framework for training a global model $f_w$ on decentralized data owned by $N$ different users $\{u_j\}_{j=1}^N$.  A central server coordinates the training of $f_w$ by iteratively aggregating gradients
computed locally by the users. Let $t \in [0, T]$ be the current iteration of the FL protocol, and $w^t$ be the set of parameters at iteration $t$. At iteration $t=0$, the global $w^t$ is initialized at random by a central server. At every iteration $t$, a subset of $M<N$ users is randomly selected to contribute to the training. Each of the selected users $u_j$ obtains $f_{w^t}$ from the central server and calculates the gradients $G_j^t$ for $f_{w^t}$ using their local training batch $\mathcal{D}_j$. In specific, $G_j^t = \nabla_{w^t} \mathcal{L}(\mathcal{D}_j, w^t)$ where $\mathcal{L}$ is a loss function. Then, each $u_j$ uploads its gradients to the central server. With a learning rate $\eta$, the server 
averages 
these gradients to update the global model’s parameters as follows:

\begin{equation}
    G^{t} = \frac{1}{M}\sum_{j=1}^M G_j^{t}, \quad w^{t+1} = w^t - \eta G^t
\end{equation}
The training continues until $f_{w^t}$ converges.

\subsection{\textbf{Image Augmentation}}
Image augmentation is a very useful technique in deep learning 
that allows for the expansion of a training dataset with artificially generated data. Given a dataset of images, augmenting each image 
using rotation, shearing, or flipping 
yields a new expanded dataset for training.
This added preprocessing helps increase model generalization and avoid overfitting 
by altering the makeup of data and adding it to the training set \cite{ho2019population,yarats2020image}. 
With more data to train, the model is less prone to \textit{memorize} the data, but generalize the pattern between the data. 
In turn, increased model generalization \textit{tends to} lead to higher model performance. Image augmentation has been widely used in 
datasets like ImageNet \cite{krizhevsky2012imagenet} and CIFAR-10 \cite{cubuk2019autoaugment}. 

Our work focuses on three main transformations: rotation, shearing, and flipping. In each augmentation scenario, we consider a 2D image $I$ where $I(i,j)$ denotes the pixel value at coordinates $(i,j)$. Rotation includes tilting an image's pixels by an angle $\theta$. We define major rotation angles as the maximum degrees of each respective quadrant in an x-y coordinate system (i.e., $90^\circ$, $180^\circ$, and $270^\circ$). Minor rotation angles are described as any angle $<90^\circ$. 
More formally, 
an image $I'$ can be constructed from $I$ as follows:
\begin{equation}\label{equ:rotate}
    I'(i, j) = I(i\cos(\theta) - j\sin(\theta), i\sin(\theta) + j\cos(\theta)) \quad \forall i,j
\end{equation}
where $\theta$ is the angle in which an image is rotated.

Flipping includes reflecting an image on its x-axis (vertical flip) or its y-axis (horizontal flip). A horizontally flipped image $I'$ can be constructed from $I$ as follows:
\begin{equation}\label{equ:hflip}
    I'(i,j) = I(-i,j) \quad \forall i,j
\end{equation}

\noindent Similarly, a vertically flipped image $I'$ can be constructed from $I$ as follows:
\begin{equation}\label{equ:vflip}
    I'(i,j) = I(i,-j) \quad \forall i,j
\end{equation}

Shearing is projecting a point or set of points within an image in a different direction. A sheared image $I'$ can be constructed from $I$ as follows:
\begin{equation}\label{equ:shear}
    I'(i,j) = I(i + \mu j, j) \quad \forall i,j
\end{equation}
where $\mu$ is the shear factor controlling the \textit{shearing intensity}.

\section{OASIS -- A Proposed Defense}\label{sec:defense}
This section describes our proposed defense, OASIS, against \textit{active} reconstruction attacks via gradient inversion in FL. In order to devise an effective defense, we analyze the attack surface to determine the core vulnerability of the system and how the attacks exploit it in principle. We then propose OASIS to prevent such exploitation, effectively tackling this class of 
attacks, regardless of how they are implemented.

\subsection{\textbf{Generalizing Active Reconstruction Attacks via Gradient Inversion}}\label{sec:threat}
\noindent \textbf{Threat Model.} We examine a 
server that is dishonest and 
aims to reconstruct the private data of a targeted user. As discussed in previous work \cite{boenisch2023curious,fowl2021robbing}, a dishonest server is capable of making malicious modifications to $w$ before dispatching it to the users at any iterations. These modifications can include changing and/or adding model parameters. However, the modification should be minimal to avoid detection.

For this attack, the adversary places a malicious fully-connected layer consisting of $n$ attacked neurons in the neural network model $f_w$, 
 so that inverting the gradients of these neurons would recover the users' data. Generally, the attack becomes less effective when the layer is placed deeper in the neural network. For the purpose of devising a robust defense, we consider a strong adversary who can place the malicious layer directly right after the input layer. The layer is parameterized by a weight matrix $W \in \mathbb{R}^{n\times d}$ and a bias vector $b \in \mathbb{R}^{n}$. Denoting $\mathcal{D} = \{x_j \in R^{d}\}_{j=1}^{B}$ as the local training data of a targeted user where $B$ is the batch size, the goal of the attack is to reconstruct the data points in $\mathcal{D}$ via the malicious layer. Our defense, OASIS, aims to minimize the quality of reconstruction, regardless of how the malicious layer $(W,b)$ is determined by the adversary.

\noindent \textbf{Attack Vector Analysis.} We aim to generalize 
state-of-the-art \textit{active} reconstruction attacks by deducing their core principle.  Suppose that the malicious layer is updated based on one single-input $x_t \in \mathbb{R}^d$, for each neuron $i$, the gradients of the loss with respect to the weights, and biases will be $$\left(\frac{\partial \mathcal{L}_t}{\partial W_i},\frac{\partial \mathcal{L}_t}{\partial b_i}\right)$$  where $\mathcal{L}_t$ is shorthand for $\mathcal{L}(x_t, (W, b))$. All the gradients $$\left\{\left(\frac{\partial \mathcal{L}_t}{\partial W_i},\frac{\partial \mathcal{L}_t}{\partial b_i}\right)\right\}_{i=1}^n$$ are then uploaded to the server. As shown in \cite{geiping2020inverting,fowl2021robbing,boenisch2023curious}, with a ReLU activation function, the server can perfectly reconstruct $x_t$ by dividing the gradients
as follows:
\begin{equation}\label{eq:recover1}
    \left(\frac{\partial \mathcal{L}_t}{\partial b_i}\right)^{-1} \frac{\partial \mathcal{L}_t}{\partial W_i} = x_t
\end{equation}
where $i$ is the index of a neuron that is \textit{activated} by the input $x_t$ and $\frac{\partial \mathcal{L}_t}{\partial b_i}\neq 0$. In other words, knowing the gradients $\left(\frac{\partial \mathcal{L}_t}{\partial W_i},\frac{\partial \mathcal{L}_t}{\partial b_i}\right)$ of a particular input sample $x_t$ allows perfect reconstruction of that sample via gradient inversion.

However, in practical FL, when the malicious layer is updated based on a batched input $\mathcal{D} = \{x_j\}_{j=1}^{B}$ where $B > 1$, all derivatives are summed over the batch dimension. In particular, the gradients of the malicious layer that the server receives will instead be: $$\left\{\left(\sum_{j = 1}^B \frac{\partial \mathcal{L}_j}{\partial W_i}, \sum_{j = 1}^B\frac{\partial \mathcal{L}_j}{\partial b_i}\right)\right\}_{i=1}^n$$ 
When the server performs the same inversion computation as Equation (\ref{eq:recover1}) on this summed gradient, it will reconstruct
\begin{equation*}
    \left(\sum_{j = 1}^B \frac{\partial \mathcal{L}_j}{\partial b_i}\right)^{-1} \left(\sum_{j = 1}^B \frac{\partial \mathcal{L}_j}{\partial W_i}\right)
\end{equation*}
which is 
proportional to a linear combination of the samples that activated neuron $i$. The coefficient for each sample in the linear combination depends on how much the sample contributes to the 
loss $\mathcal{L}$. Reconstructing such a combination may not be able to reveal the content of each individual input sample, thereby hindering the impact of the attack.

To circumvent the problem of summed gradients, the CAH attack proposed by \cite{boenisch2023curious} chooses the parameters for $(W,b)$ that maximize the likelihood that each attacked neuron is activated by only one sample in the batch.  The rationale behind this is that if 
$i$ is activated only by one data point $x_t$, 
then $$\left(\sum_{j = 1}^B \frac{\partial \mathcal{L}_j}{\partial W_i}, \sum_{j = 1}^B\frac{\partial \mathcal{L}_j}{\partial b_i}\right) = \left(\frac{\partial \mathcal{L}_t}{\partial W_i},\frac{\partial \mathcal{L}_t}{\partial b_i}\right)$$ since $\frac{\partial \mathcal{L}_j}{\partial W_i} = 0$ and $\frac{\partial \mathcal{L}_j}{\partial b_i} = 0$ for data points $x_j (j\neq t)$ that do not activate the neuron $i$. After obtaining $\left(\frac{\partial \mathcal{L}_t}{\partial W_i},\frac{\partial \mathcal{L}_t}{\partial b_i}\right)$, the server can 
reconstruct $x_t$ by Equation (\ref{eq:recover1}). 

On the other hand, \cite{fowl2021robbing} proposes the RTF attack in which the reconstruction can be carried out by considering the difference between two successive neurons' gradients, with respect to some specific parameters $(W,b)$. Specifically, the server can strategically choose $(W,b)$ so that, given the gradients $$\left(\sum_{j = 1}^B \frac{\partial \mathcal{L}_j}{\partial W_i}, \sum_{j = 1}^B\frac{\partial \mathcal{L}_j}{\partial b_i}\right)$$ of neuron $i$ and $$\left(\sum_{j = 1}^B \frac{\partial \mathcal{L}_j}{\partial W_{i+1}}, \sum_{j = 1}^B\frac{\partial \mathcal{L}_j}{\partial b_{i+1}}\right)$$ of neuron $i+1$, the difference between them can reveal the gradients $\left(\frac{\partial \mathcal{L}_t}{\partial W_i}, \frac{\partial \mathcal{L}_t}{\partial b_i}\right)$ of a particular sample $x_t$ that activates neuron $i$. With this, Equation (\ref{eq:recover1}) 
can perfectly reconstruct that sample $x_t$.

From this analysis, we can observe the \textit{underlying principle of these attacks}: as long as the gradients $\left(\frac{\partial \mathcal{L}_t}{\partial W_i}, \frac{\partial \mathcal{L}_t}{\partial b_i}\right)$ of one individual sample $x_t$ can be extracted from the summed gradients $$\left\{\left(\sum_{j = 1}^B \frac{\partial \mathcal{L}_j}{\partial W_i}, \sum_{j = 1}^B\frac{\partial \mathcal{L}_j}{\partial b_i}\right)\right\}_{i=1}^n$$ with $\frac{\partial \mathcal{L}_t}{\partial b_i} \neq 0$, that sample $x_t$ can be perfectly reconstructed by gradient inversion via Equation (\ref{eq:recover1}). Therefore, the attack strategies specifically involve choosing $(W,b)$ that optimizes the chance of 
extraction, thus improving reconstruction quality.

\noindent \textbf{Defense Intuition.} By this principle, to effectively defend against such attacks, it is essential to prevent the leaking of any individual data points' gradients from the summed gradients, regardless of how the parameters $(W,b)$ are chosen. With this in mind, we establish the following proposition:

\begin{proposition}\label{prop1}
Given a sample $x_t\in \mathcal{D}$, if there exists an $x'_t\in \mathcal{D}$ such that $x_t$ and $x'_t$ activate the same set of neurons in the malicious layer, then the adversary cannot extract $$\left(\frac{\partial \mathcal{L}_t}{\partial W_i}, \frac{\partial \mathcal{L}_t}{\partial b_i}\right)$$ with $\frac{\partial \mathcal{L}_t}{\partial b_i} \neq 0$ from $$\left\{\left(\sum_{j = 1}^B \frac{\partial \mathcal{L}_j}{\partial W_i}, \sum_{j = 1}^B\frac{\partial \mathcal{L}_j}{\partial b_i}\right)\right\}_{i=1}^n$$ 
\end{proposition}
\begin{proof}
    There are two cases in which the adversary is able to obtain $\left(\frac{\partial \mathcal{L}_t}{\partial W_i}, \frac{\partial \mathcal{L}_t}{\partial b_i}\right)$ with $\frac{\partial \mathcal{L}_t}{\partial b_i} \neq 0$ from $$\left\{\left(\sum_{j = 1}^B \frac{\partial \mathcal{L}_j}{\partial W_i}, \sum_{j = 1}^B\frac{\partial \mathcal{L}_j}{\partial b_i}\right)\right\}_{i=1}^n$$

\noindent \textbf{(1)} There exists an $i \in \{1,2,...,n\}$ s.t. $$\left(\sum_{j = 1}^B \frac{\partial \mathcal{L}_j}{\partial W_i}, \sum_{j = 1}^B\frac{\partial \mathcal{L}_j}{\partial b_i}\right) = \left(\frac{\partial \mathcal{L}_t}{\partial W_i}, \frac{\partial \mathcal{L}_t}{\partial b_i}\right)$$ This means that the neuron $i$ is activated only by $x_t$, thus contradicting the fact that $x_t$ and $x'_t$ activate the same set of neurons.\\
\textbf{(2)} There exists a subset $D\subseteq \mathcal{D} \setminus x_t$ such that the adversary can determine $$\left(\sum_{x_j \in D} \frac{\partial \mathcal{L}_j}{\partial W_i}, \sum_{x_j\in D}\frac{\partial \mathcal{L}_j}{\partial b_i}\right)$$ and $$\left(\sum_{x_j \in D \cup x_t} \frac{\partial \mathcal{L}_j}{\partial W_i}, \sum_{x_j\in D \cup x_t}\frac{\partial \mathcal{L}_j}{\partial b_i}\right)$$ from $$\left\{\left(\sum_{j = 1}^B \frac{\partial \mathcal{L}_j}{\partial W_i}, \sum_{j = 1}^B\frac{\partial \mathcal{L}_j}{\partial b_i}\right)\right\}_{i=1}^n$$ 
To be able to obtain $$\left(\sum_{x_j \in D \cup x_t} \frac{\partial \mathcal{L}_j}{\partial W_i}, \sum_{x_j\in D \cup x_t}\frac{\partial \mathcal{L}_j}{\partial b_i}\right)$$ it must be that $x_t$ activates neuron $i$. This also means that $x'_t$ activates neuron $i$ (since $x_t$ and $x'_t$ activate the same set of neurons)  and that $x'_t \in D$. But in order to get $$\left(\sum_{x_j \in D} \frac{\partial \mathcal{L}_j}{\partial W_i}, \sum_{x_j\in D}\frac{\partial \mathcal{L}_j}{\partial b_i}\right)$$ there must be a neuron that is activated by samples in $D$, which includes $x'_t$, and is not activated by $x_t$. This 
contradicts the fact that $x_t$ and $x'_t$ activate the same set of neurons.
\end{proof}
Intuitively, suppose that for every $x_t \in \mathcal{D}$, we find a data point $x'_t$ such that $x_t$ and $x'_t$ always activate the same set of neurons, and then we add $x'_t$ to $\mathcal{D}$. From Proposition \ref{prop1}, it can be inferred that the best that the attacker can do is extracting $$\left(\frac{\partial \mathcal{L}_t}{\partial W_i} + \frac{\partial \mathcal{L}_t'}{\partial W_i},\frac{\partial \mathcal{L}_t}{\partial b_i} + \frac{\partial \mathcal{L}_t'}{\partial b_i}\right)$$ from the summed gradients $$\left\{\left(\sum_{j = 1}^B \frac{\partial \mathcal{L}_j}{\partial W_i}, \sum_{j = 1}^B\frac{\partial \mathcal{L}_j}{\partial b_i}\right)\right\}_{i=1}^n$$ 
Hence, it could only reconstruct a linear combination of $x_t$ and $x'_t$. If the linear combination does not reveal the content of $x_t$, then the proposed defense is successful.

\subsection{\textbf{Image Augmentation as a Defense}}\label{sec:oasis}
From the previous attack principle and defense intuition, we devise a robust defense mechanism as follows. For every $x_t\in \mathcal{D}$, we find a set of data points $X'_t$ such that $x_t$ and every $x'\in X'_t$ activate the same set of neurons. Then, we construct a new local training dataset:
\begin{equation}\label{equ:oasis}
    \mathcal{D}' = \mathcal{D} \cup \bigcup_{t = 1}^B X'_t
\end{equation}
If $\mathcal{D}$ is labeled then the data points in $X'_t$ are given the same label as $x_t$. The user will use $\mathcal{D}'$ instead of $\mathcal{D}$ for the FL process, so that an \textit{active} reconstruction attack can only reconstruct a linear combination of $x_t$ and $x'\in X'_t$. This mechanism is illustrated in Figure \ref{fig:system-design}. The defense is considered effective if it satisfies two conditions: (1) using $\mathcal{D}'$ does not heavily reduce the training performance, 
and (2) a linear combination of $x_t$ and $x'\in X'_t$ 
does not reveal the content of $x_t$.

To find $X'_t$ that activates the same set of neurons as $x_t$, we propose using image augmentation \cite{yarats2020image} where $X'_t$ contains the transformations of $x_t$, such as rotation, shearing, or flipping. 
As noted in \cite{cubuk2019autoaugment}, image augmentation can be 
used to teach a model about invariances in the data domain. For that reason, training with image augmentation makes the model invariant to the transformations of images. In other words, the model should exhibit similar behavior (i.e., similar patterns of neuron activations) given different transformations of an image. As a result, $x_t$ and images in $X'_t$ are \textit{more likely} to activate the same set of neurons. Our experiments in Section 
4, especially Figures \ref{fig:og_major_dataset}-\ref{fig:combo}, further support this claim by showing that the reconstructed image is a linear combination of the transformed and the original, which is caused by $x_t$ and $X'_t$ activating the same set of
neurons.

Furthermore, using image augmentation as a defense also satisfies the above-mentioned two conditions. First, using image augmentation 
maintains the training performance as it was originally designed to improve model generalization and reduce overfitting. Second, as we shall demonstrate in Section 
4, a linear combination of an image $x_t$ and its transformations yields an unrecognizable image, thereby protecting the original content of $x_t$.


\section{Experimental Analysis}\label{sec:exp}


This section evaluates the performance of our defense with various experiments to shed light on how OASIS can \textit{offset} state-of-the-art \textit{active} reconstruction attacks while still maintaining the model training performance.

\begin{figure}
    \small
    \centering
    \begin{tabular}{lcccccc}
        \includegraphics[width=2cm]{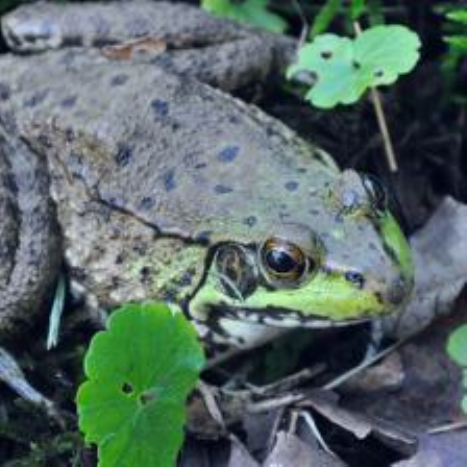}&
        \includegraphics[width=2cm]{examples/redo-no-defense.pdf}&
        \includegraphics[width=2cm,trim={16cm 0 8cm 0},clip=True]{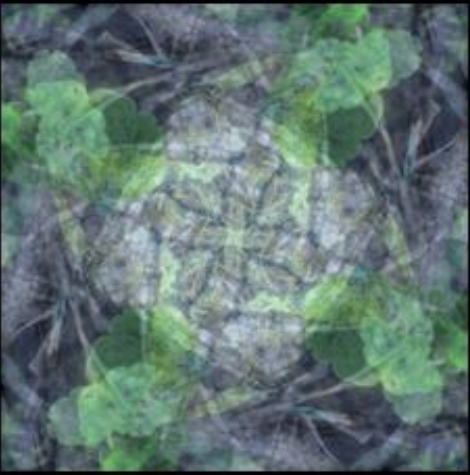}&
        \\
        & 139.17 dB & 15.41 dB
        \\
        \makecell{Original\\Image} & \makecell{Reconstruction \\ w/o OASIS} & \makecell{Reconstruction \\ with OASIS}
    \end{tabular}
    \caption{Example visual representation of PSNR values. Images with lower PSNR \textit{tend to} have worse reconstruction quality compared to images with higher PSNR.}
    \label{fig:ig-imagenet-examples}
\end{figure}

\begin{figure*}[htp!]
    \centering
    \begin{subfigure}{0.49\linewidth}
    \includegraphics[width=94mm]{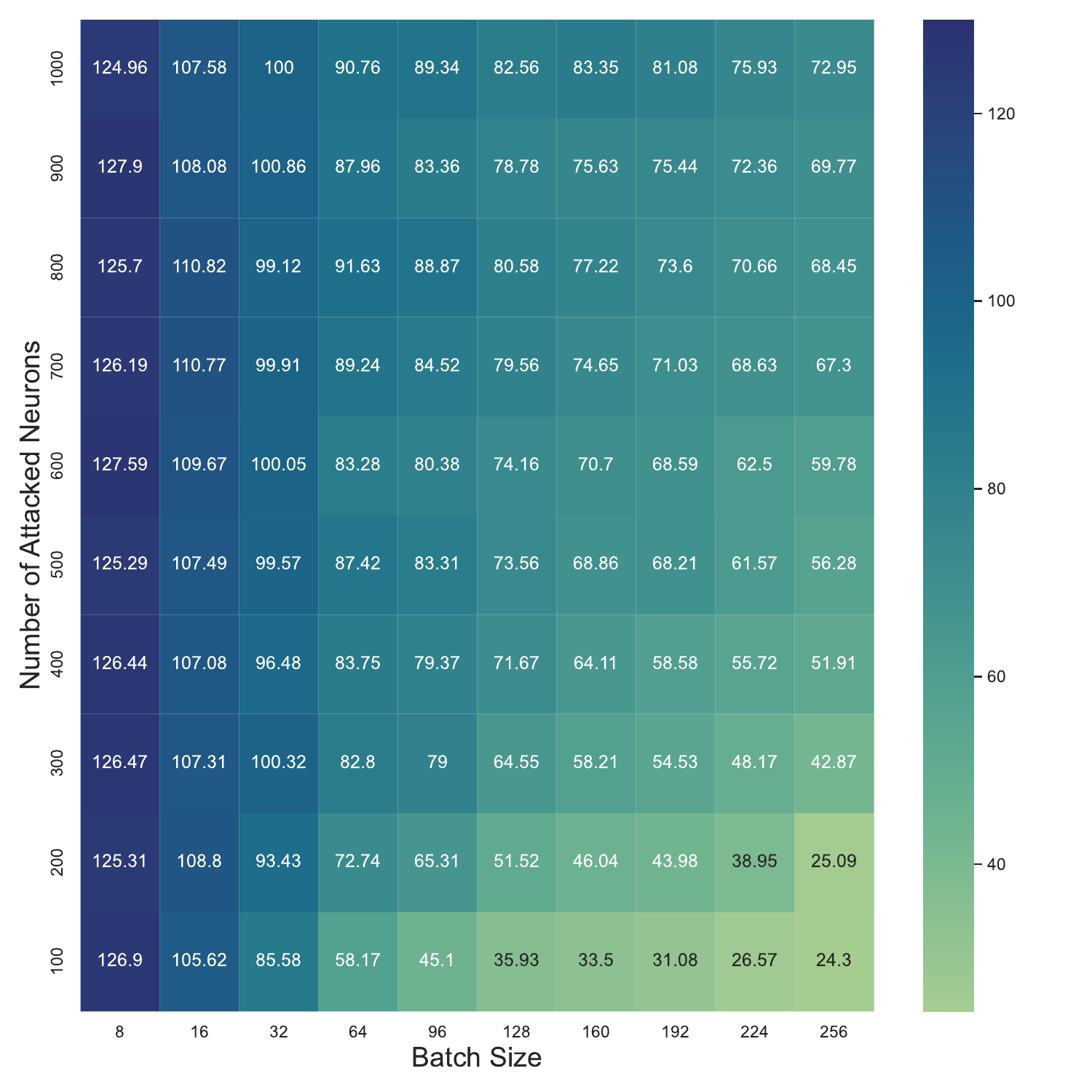}
    \caption{ImageNet}
    \end{subfigure}
    \begin{subfigure}{0.49\linewidth}
    \includegraphics[width=94mm]{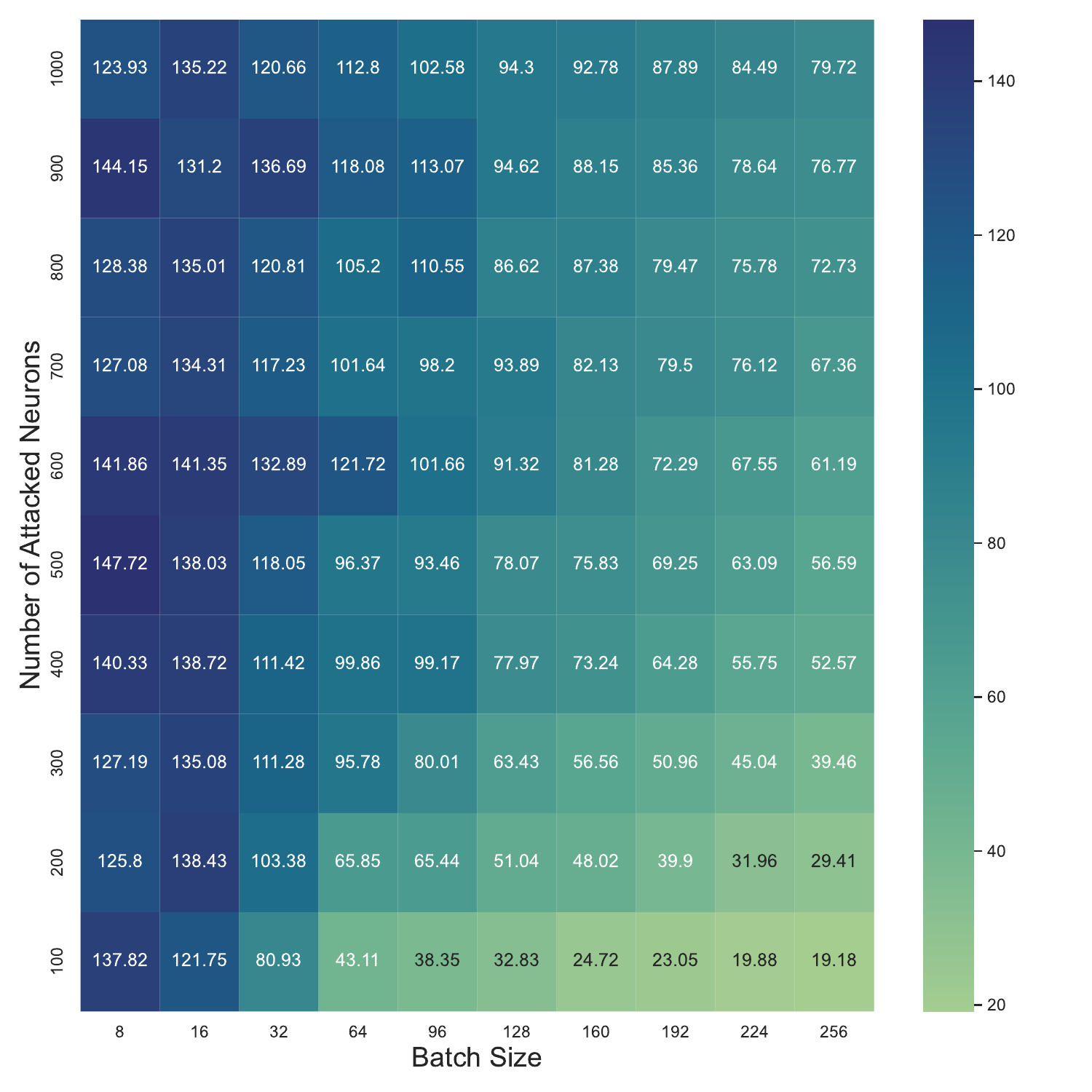}
    \caption{CIFAR100}
    \end{subfigure}
    \caption{Average PSNR over the images reconstructed by the RTF attack w.r.t the batch size and the number of attacked neurons on ImageNet and CIFAR100.}
    \label{fig:rtf-config}
\end{figure*}

\begin{figure*}[htp!]
    \centering
    \begin{subfigure}{0.49\linewidth}
    \includegraphics[width=94mm]{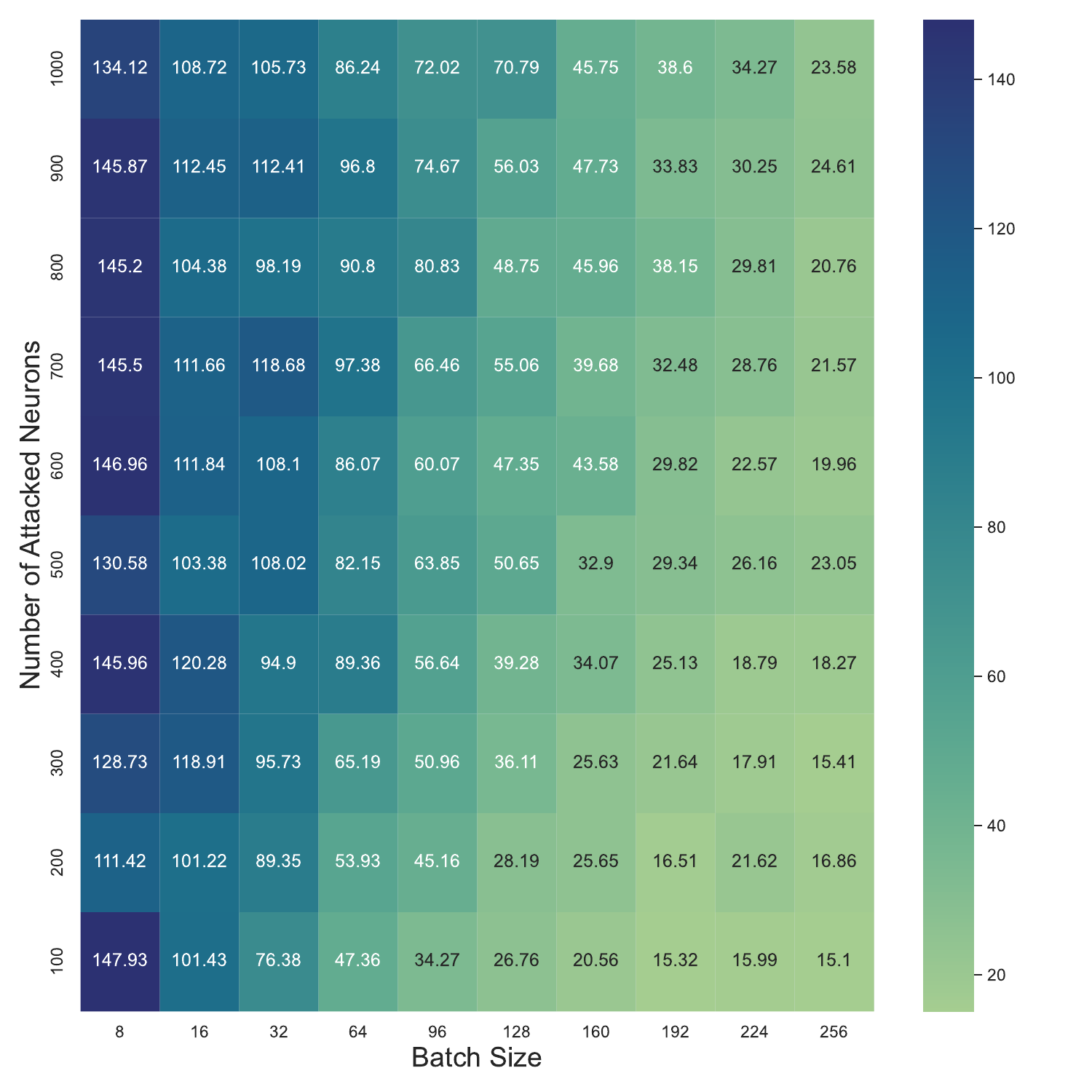}
    \caption{ImageNet}
    \end{subfigure}%
    \begin{subfigure}{0.49\linewidth}
    \includegraphics[width=94mm]{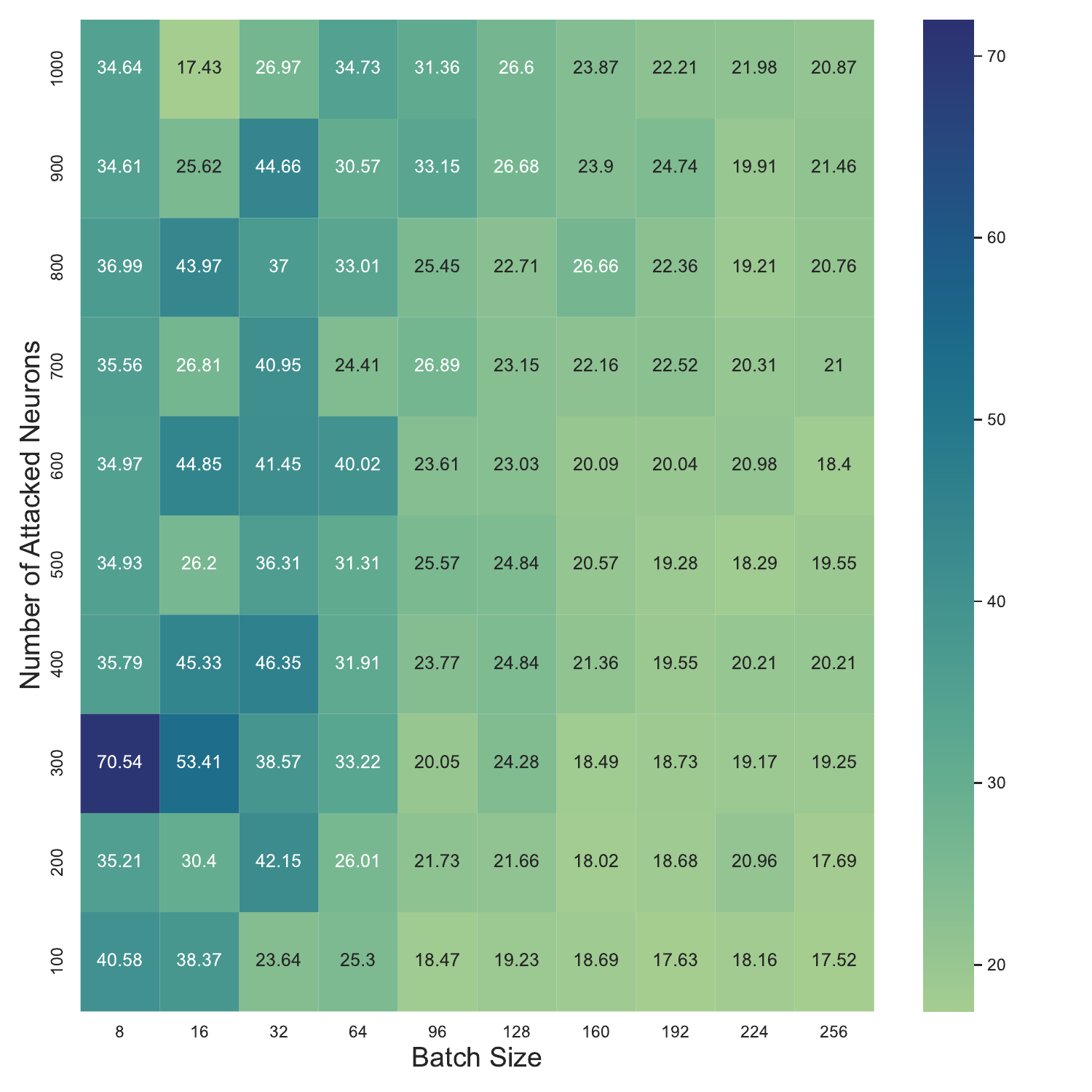}
    \caption{CIFAR100}
    \end{subfigure}
    \caption{Average PSNR over the images reconstructed by the CAH attack w.r.t the batch size and the number of attacked neurons on ImageNet and CIFAR100.}
    \label{fig:cah-config}
\end{figure*}

\subsection{\textbf{Experimental Settings}} \label{sec:settings}
We conduct two state-of-the-art \textit{active} reconstruction attacks, namely \textit{Robbing the Fed (RTF)} \cite{fowl2021robbing} and \textit{Curious Abandon Honesty (CAH)} \cite{boenisch2023curious}, against our OASIS defense on two datasets  ImageNet \cite{deng2009imagenet} and CIFAR100 \cite{krizhevsky2009learning}. 
For these attacks, we adopt the implementation from \url{https://github.com/JonasGeiping/breaching}. 
To capture how OASIS mitigates the success rate of the attacks, similar to previous work \cite{geiping2020inverting,fowl2021robbing}, we use the \textit{Peak Signal-to-Noise Ratio (PSNR)} value to measure the quality of a reconstructed image with respect to the original image.  Higher PSNR values indicate better reconstruction quality, thus higher attack success rates. Figure \ref{fig:ig-imagenet-examples} illustrates a visual representation of PSNR values. Our goal is to minimize the PSNR values of reconstructed images. Furthermore, we visually compare the reconstructed images when using OASIS against their respective original images to demonstrate how OASIS protects the content of the dataset. Finally, we measure model performance for each augmentation method on each dataset. OASIS is expected to impose a negligible trade-off on the performance of training models.

\begin{figure*}[htp!]
    \centering
    \begin{subfigure}{0.49\linewidth}
        \centering
        \includegraphics[width=0.49\textwidth]{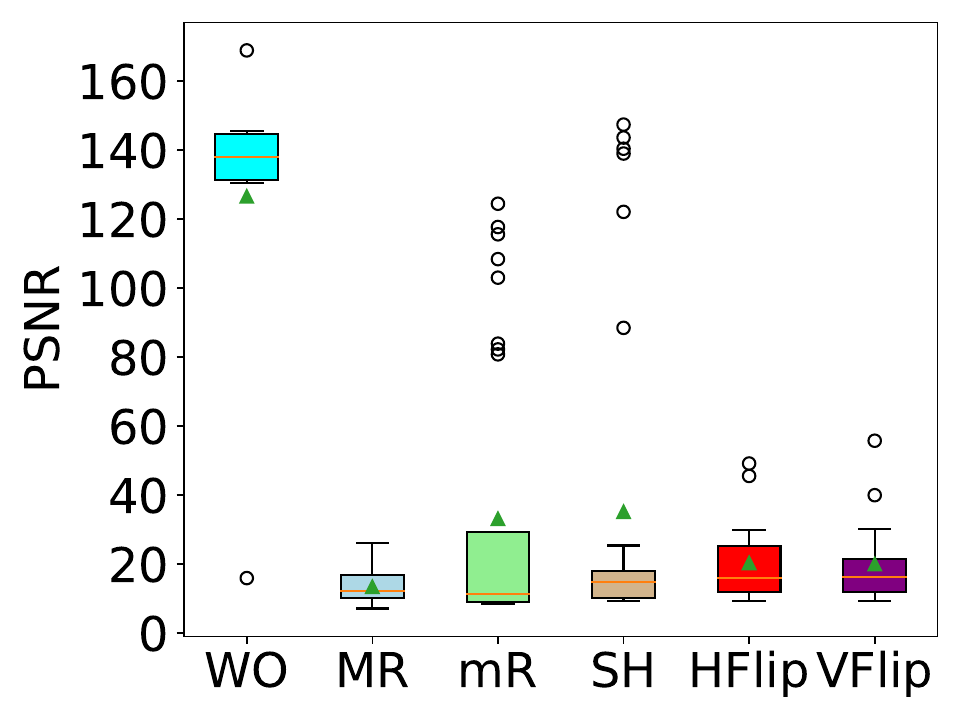}
        \includegraphics[width=0.49\textwidth]{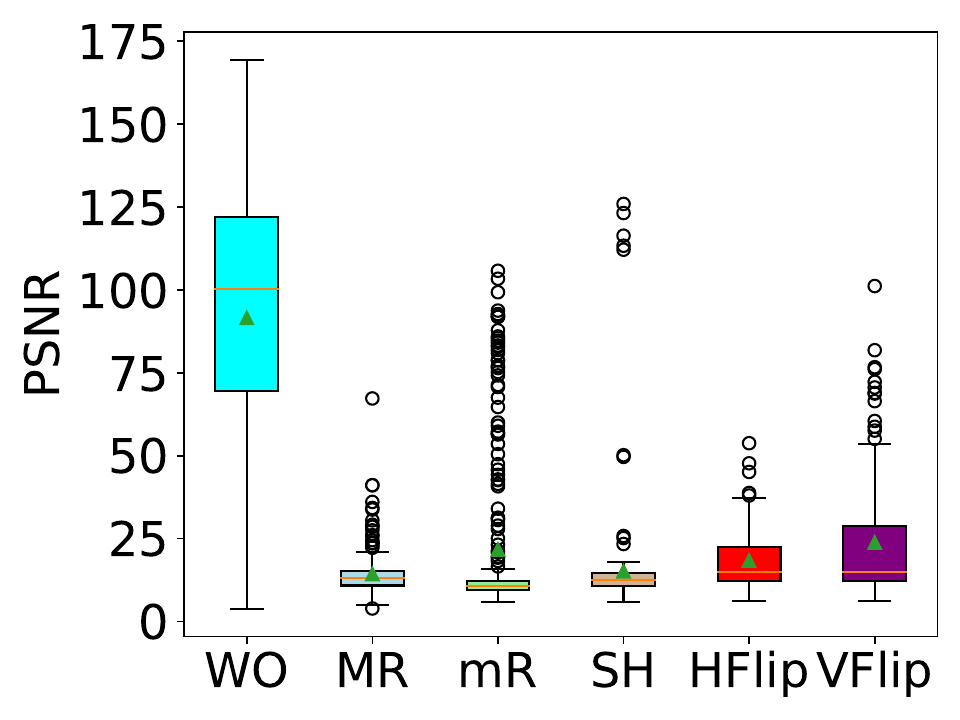}
        \caption{ImageNet. \textit{Left:} $(B,n) = (8,900)$. \textit{Right:} $(B,n) = (64,800)$}
    \end{subfigure}
    \hfill
    \begin{subfigure}{0.49\linewidth}
        \centering
        \includegraphics[width=0.49\textwidth]{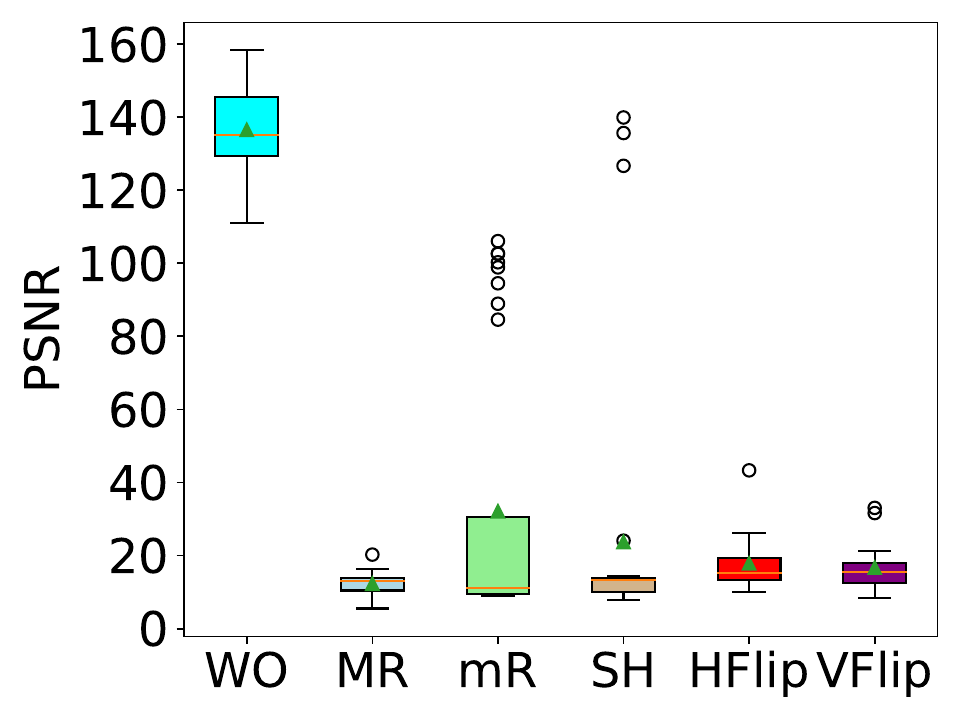}
        \includegraphics[width=0.49\textwidth]{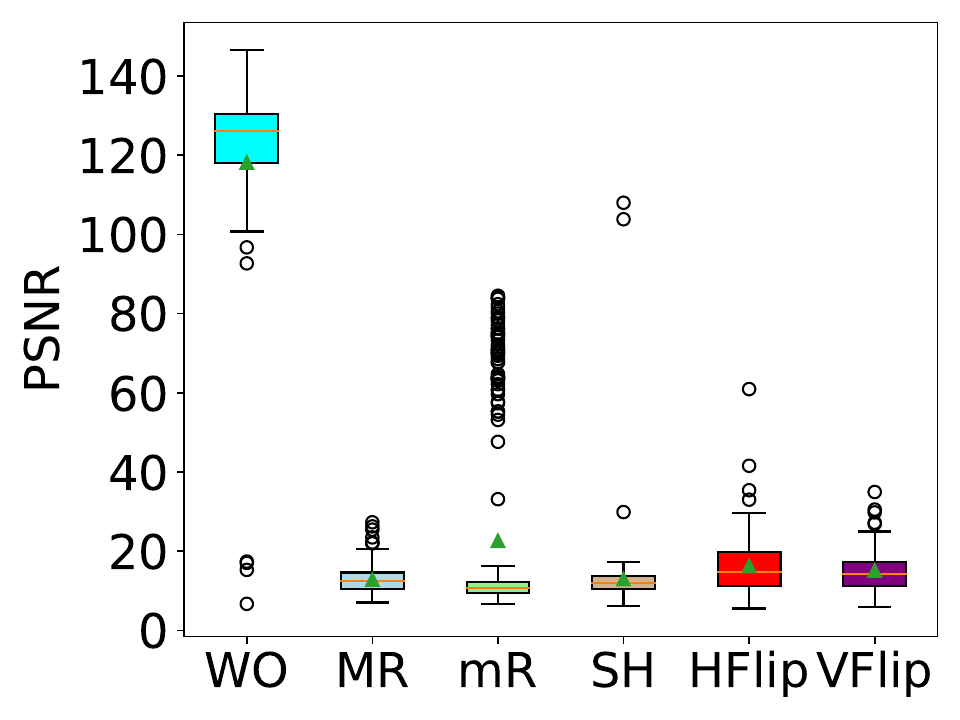}
        \caption{CIFAR100. \textit{Left:} $(B,n) = (8,500)$. \textit{Right:} $(B,n) = (64,600)$}
    \end{subfigure}
    \caption{PSNR values of images reconstructed by the RTF attack w.r.t different transformations and different batch sizes on ImageNet and CIFAR100. The green triangle denotes the average PSNR over all reconstructed images. \emph{(WO = Without OASIS, MR = Major Rotation, mR = Minor Rotation, SH = Shearing, HFlip = Horizontal Flip, and VFlip = Vertical Flip)}}
    \label{fig:rtf-psnr}
\end{figure*}

\begin{figure*}[htp!]
    \centering
    \begin{subfigure}{0.49\linewidth}
        \centering
        \includegraphics[width=0.49\textwidth]{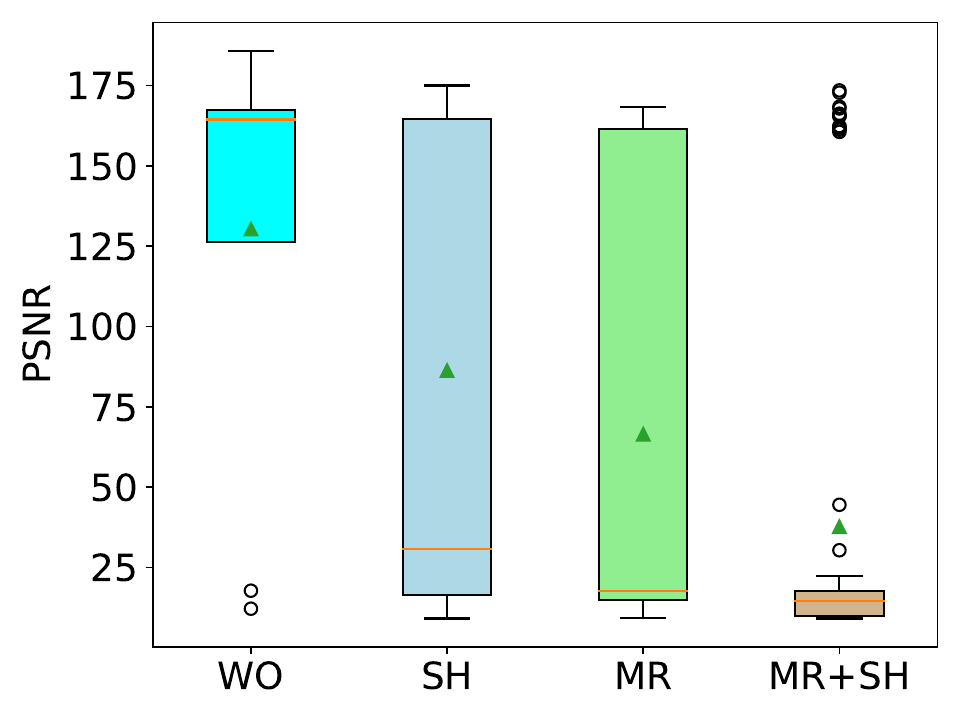}
        \includegraphics[width=0.49\textwidth]{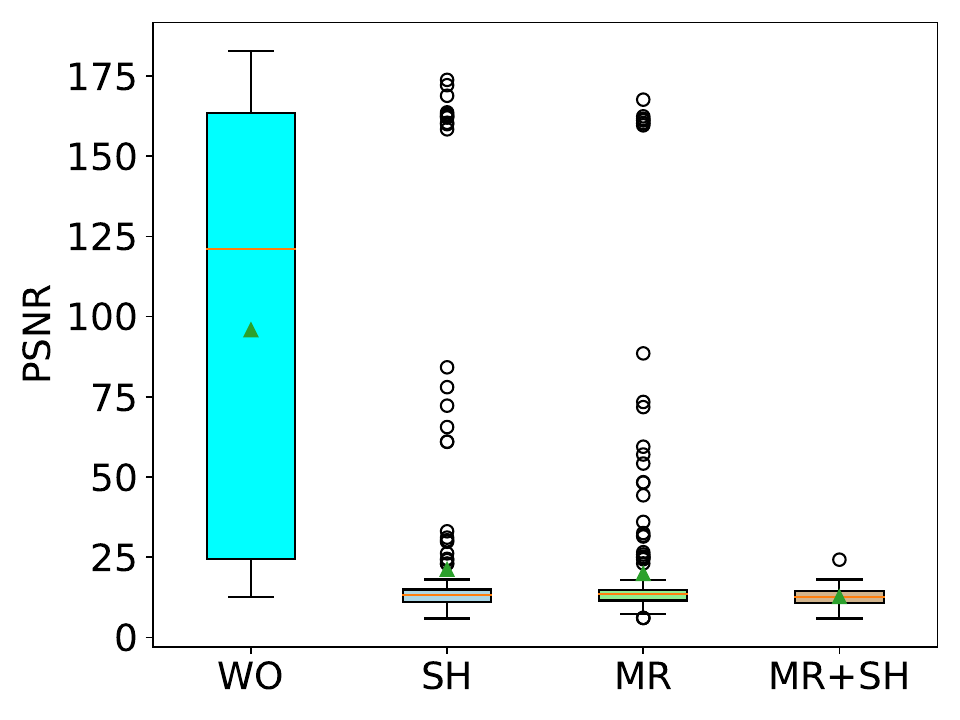}
        \caption{ImageNet. \textit{Left:} $(B,n) = (8,100)$. \textit{Right:} $(B,n) = (64,700)$}
        \label{fig:cah-psnr-imgnet}
    \end{subfigure}
    \hfill
    \begin{subfigure}{0.49\linewidth}
        \centering
        \includegraphics[width=0.49\textwidth]{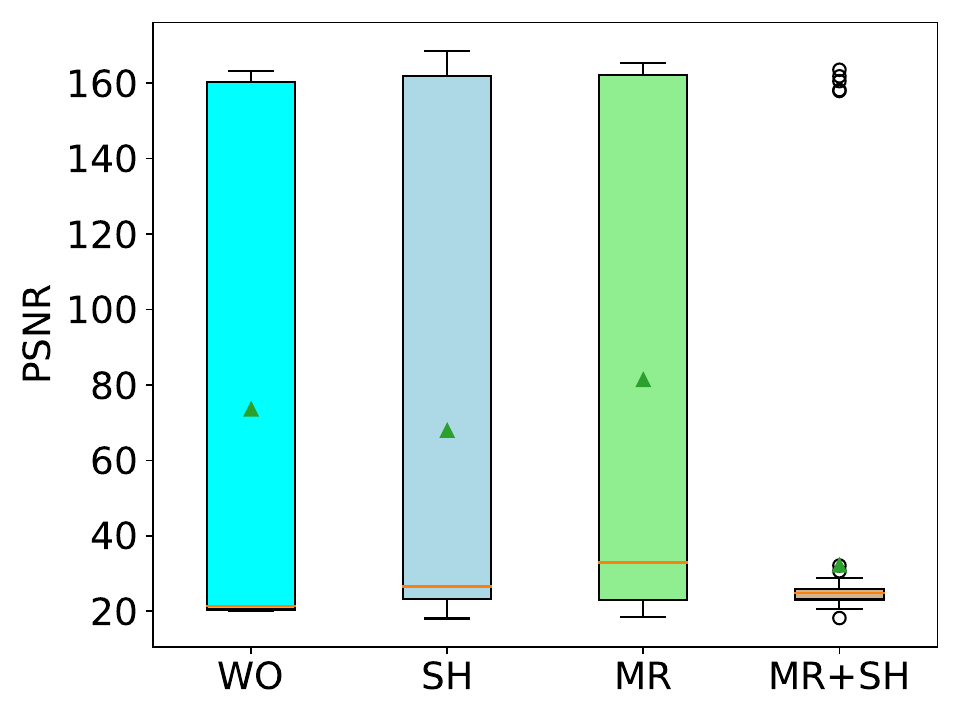}
        \includegraphics[width=0.49\textwidth]{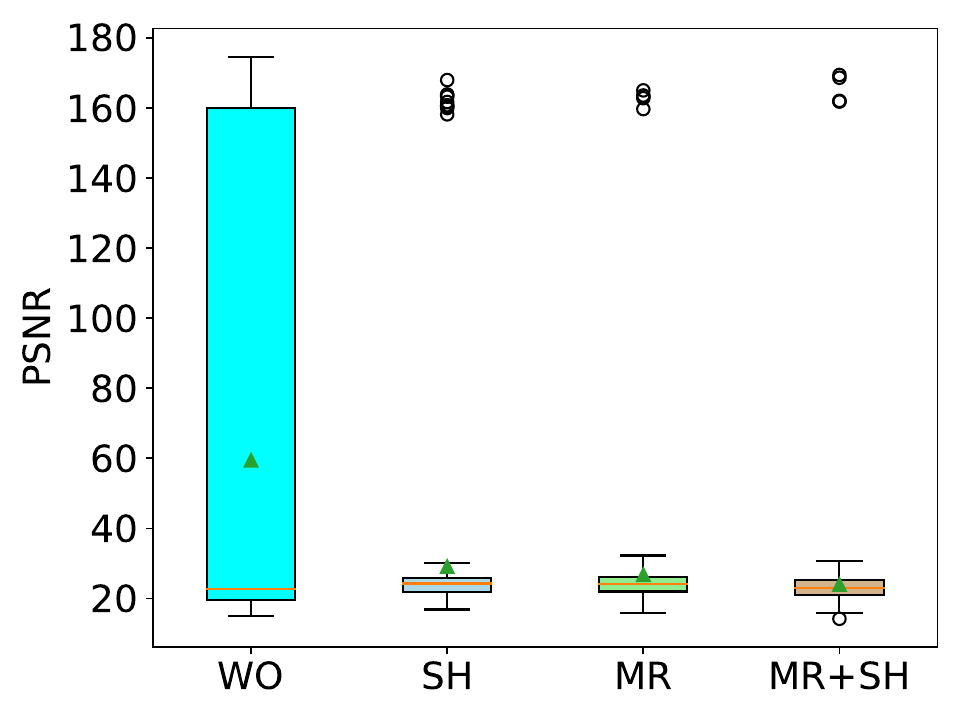}
        \caption{CIFAR100. \textit{Left:} $(B,n) = (8,300)$. \textit{Right:} $(B,n) = (64,600)$}
        \label{fig:cah-psnr-cifar}
    \end{subfigure}
    \caption{PSNR values of images reconstructed by the CAH attack w.r.t different transformations and different batch sizes on ImageNet and CIFAR100. The green triangle denotes the average PSNR over all reconstructed images. \emph{(WO = Without OASIS, SH = Shearing, MR = Major Rotation, and MR + SH = Major Rotation + Shearing)}}
    \label{fig:cah-psnr}
\end{figure*}


For a fair evaluation, the attacks are first configured to have the highest success rate. 
As discussed in the threat model in Section \ref{sec:threat}, the malicious layer is appended right after the input layer. 
Furthermore, the attack performance depends on the number of attacked neurons $n$, and the batch size $B$.  Generally, it is straightforward that the reconstruction attacks perform worse with larger batch sizes. We experiment with two batch sizes: $B=8$ for evaluating against strong attacks, and $B=64$ for a more realistic training configuration. 
We conduct preliminary experiments to find the hyperparameters that result in the strongest attacks. Specifically, we test the attacks with various batch sizes and numbers of attacked neurons, and report the average PSNR value over the images reconstructed by RTF and CAH in Figures \ref{fig:rtf-config} and \ref{fig:cah-config}, respectively. 
As previously stated, the reconstruction attacks perform worse with larger batch sizes, and that behavior is illustrated in Figures \ref{fig:rtf-config} and \ref{fig:cah-config}. 
For each batch size, we choose the number of attacked neurons $n$ that yields the highest average PSNR.

As can be seen in Figure \ref{fig:rtf-config}, the RTF attack's optimal settings for ImageNet with a batch of 8 occur with 900 attacked neurons yielding an average PSNR value of 127.9 dB. The optimal settings for a batch of 64 occur with 800 attacked neurons yielding an average PSNR value of 91.63 dB. For CIFAR100, we see the optimal settings for a batch of 8 and 64 are 500 and 600 attacked neurons yielding average PSNR values of 147.72 dB and 121.72 dB, respectively. 



We test for the optimal settings of the CAH attack in a similar manner in Figure \ref{fig:cah-config}. For ImageNet, a batch of 8 with 100 attacked neurons produces an average PSNR value of 147.93 dB and a batch of 64 with 700 attacked neurons produces an average PSNR value of 97.38 dB. CIFAR100 was treated the same as before. A batch of 8 along with 300 attacked neurons results in an average PSNR value of 70.54 dB while a batch of 64 with 600 attacked neurons yields an average PSNR value of 40.02 dB.

\noindent\textbf{OASIS Implementation.} As for constructing $\mathcal{D}'$ in Equation \ref{equ:oasis}, we test with various methods of image augmentation, including rotation, shearing, and flipping, and observe how each of them impacts the performance of OASIS. We describe how the transformations are implemented as follows. In the case of major rotation, every image in $\mathcal{D}$ was rotated three different times at angles of $90^\circ$, $180^\circ$, and $270^\circ$, following Equation \ref{equ:rotate}, to generate three transformed images for $\mathcal{D}'$. For minor rotation, we rotate each image three different times at angles of $30^\circ, 45^\circ, 60^\circ$.

For flipping, we conduct both horizontal and vertical flipping using Equations \ref{equ:hflip} and \ref{equ:vflip}, respectively. In regard to shearing, we follow Equation \ref{equ:shear} and shear every image in $\mathcal{D}$ with three different shear factors of 0.55, 1.0, and 0.9 to generate three transformed images for $\mathcal{D}'$. Each transformation is implemented with the official PyTorch Vision library\footnote{\textit{\url{https://github.com/pytorch/vision.git}}} and the Kornia library\footnote{\textit{\url{https://github.com/kornia/kornia.git}}}.
\subsection{\textbf{OASIS Defensive Performance}}
Figure \ref{fig:rtf-psnr} depicts the effectiveness of our defense in regard to reducing the reconstruction quality of the RTF attack. 
Five transformations are used in this experiment, and it can be seen from Figure \ref{fig:rtf-psnr} that each of them substantially reduce the PSNR values of the reconstructed images across all testing scenarios. Specifically, without OASIS, most of the images reconstructed by the RTF attack have PSNR ranging from 130 dB to 145 dB at batch size 8, indicating perfect reconstruction. Major rotation is the most robust transformation such that by adding rotations at major angles 
to each image in $\mathcal{D}$, the resulting reconstruction by RTF only yields PSNR from 15 dB to 20 dB. Thus, the content of each image in $\mathcal{D}$ remains hidden.

To understand how the major rotation can invalidate the RTF attack, we note that the activation of attacked neurons in RTF depends on a scalar quantity of the input, such as the average of pixel values \cite{fowl2021robbing}. Major rotation imposes minimal change to this quantity (it does not change the average of pixel values). Hence, using this transformation for building $X'_t$ ensures that $x_t$ and $X'_t$ activate the same set of neurons, for all $x_t\in \mathcal{D}$. Furthermore, as we shall see in Section \ref{ssec:visual}, 
a linear combination of an image and its rotations yields an unrecognizable image. 
We also note that flipping does not change the average of pixel values either, however, this transformation does not necessarily result in unrecognizable reconstruction (as shown later in Section \ref{ssec:visual}), thus its PSNR is slightly higher than that of major rotation.

Figure \ref{fig:cah-psnr} illustrates the performance of OASIS against the CAH attack. 
With batch size 64, we observe a similar result as the previous experiment against RTF in which the major rotation keeps the PSNR of reconstructed images low. However, for batch size 8, the major rotation fails to prevent many images from being perfectly reconstructed. The same behavior is exhibited through shearing. The core issue here is that these transformations alone are not enough to prevent several $x_t\in \mathcal{D}$ from being the sole activation of certain attacked neurons in CAH, thus the content of those $x_t$ is revealed through reconstruction.

To tackle this issue, we attempt to integrate multiple transformations to increase the likelihood that $x_t$ and some images in $X'_t$ activate the same set of neurons in the malicious layer. In other words, the set $X'_t$ is constructed by more than one transformation. As shown in Figure \ref{fig:cah-psnr}, we experiment with integrating the two most robust transformations: major rotation and shearing. This integration is able to render the reconstruction by CAH unrecognizable with low PSNR. Specifically, with ImageNet (Figure \ref{fig:cah-psnr-imgnet}), it significantly decreases the PSNR of reconstructed images from above 125 dB to below 25 dB. The same effect is also exhibited with CIFAR100 (Figure~\ref{fig:cah-psnr-cifar}). 


\subsection{\textbf{Visual Reconstructions}}\label{ssec:visual}
We visually demonstrate the resulting reconstruction from the attacks. The goal is to show that, with our OASIS defense, the attacks indeed reconstruct a linear combination of an image and its transformations, effectively confirming the claims in Section \ref{sec:oasis}. 
Moreover, it shows that the linear combination yields the reconstructed image unrecognizable, protecting the content of the input images. 

\noindent \textbf{\underline{\textit{Rotation.}}} Figures \ref{fig:og_major_dataset} and \ref{fig:og_minor_dataset} illustrate the reconstruction from the RTF attack with major rotation and minor rotation being used as augmentations from OASIS, respectively. We can see that the reconstructed images are 
an overlap of the original images and their respective rotations. As previously discussed in Figure \ref{fig:rtf-psnr}, major rotation is the most effective transformation with the lowest PSNR for reconstruction, and we can see in Figure \ref{fig:og_major_dataset} that the reconstructed images are unrecognizable. Although the reconstruction with minor rotation has higher PSNR, 
Figure \ref{fig:og_minor_dataset} shows that it is still challenging to discern the original images 
from the reconstructed ones.

\begin{figure}
    \centering
    \begin{tabular}{c|c}
        \includegraphics[width=40mm,trim={24cm 0 0 0},clip=True]{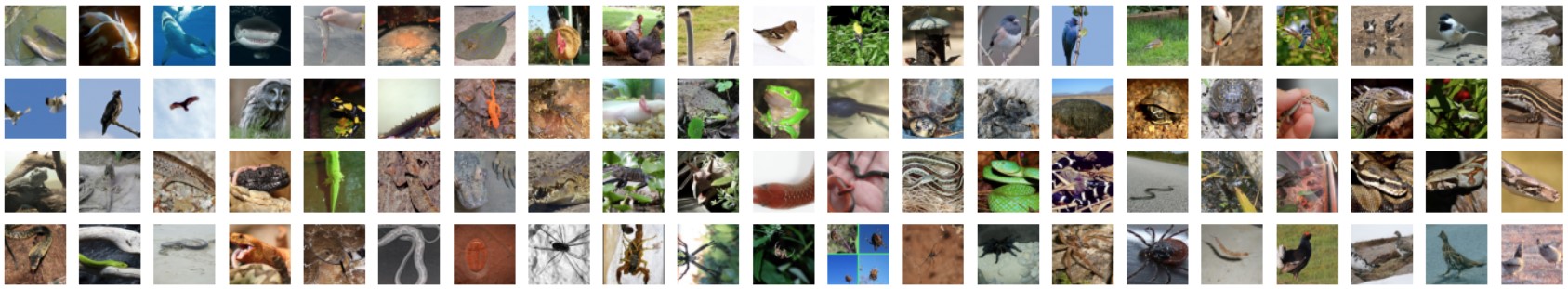}
        &
        \includegraphics[width=40mm,trim={24.20cm 0 0 0.2cm},clip=True]{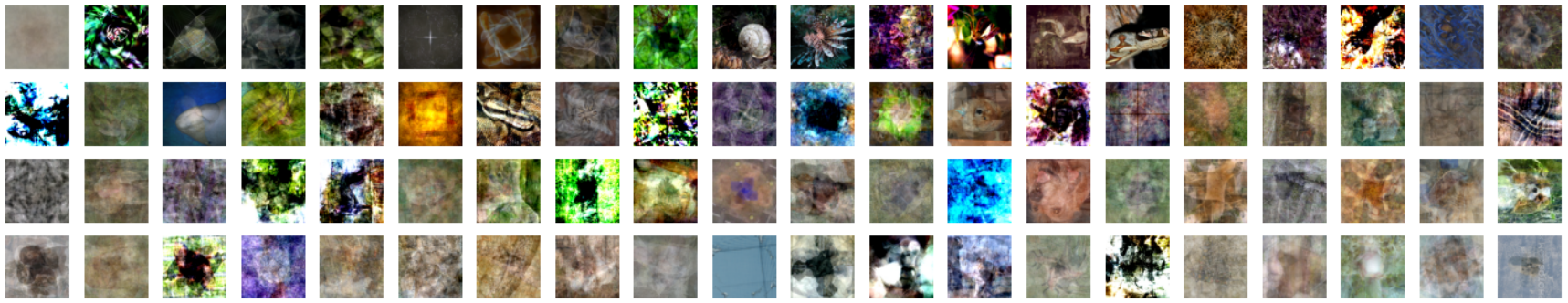}
    \end{tabular}
    \caption{\textit{Left:} Raw input  images. \textit{Right:} Reconstruction result with major rotation.}
    \label{fig:og_major_dataset}
\end{figure}

\begin{figure}
    \centering
    \begin{tabular}{c|c}
        \includegraphics[width=40mm,trim={24cm 0 0 0},clip=True]{Figures/original_dataset.jpg}
        &
        \includegraphics[width=40mm,trim={24.20cm 0 0 0.2cm},clip=True]{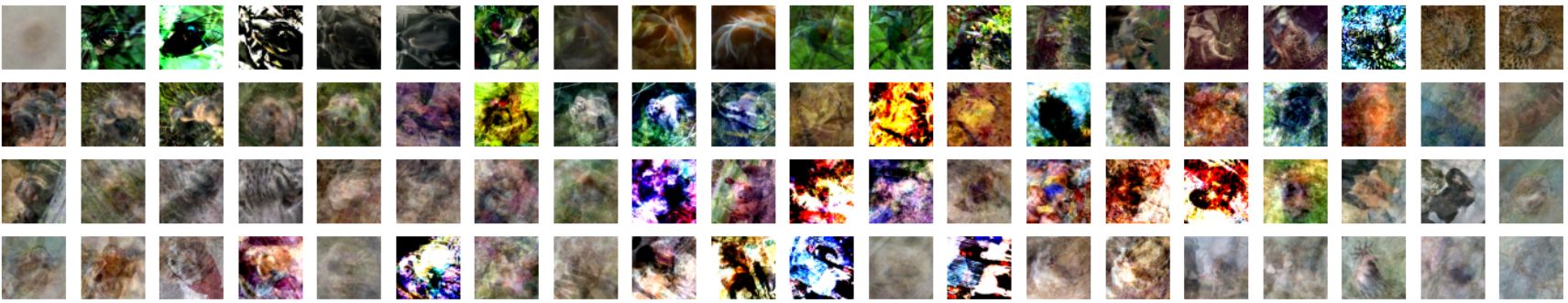}
    \end{tabular}
    \caption{\textit{Left:} Raw input  images. \textit{Right:} Reconstruction result with minor rotation.}
    \label{fig:og_minor_dataset}
\end{figure}

\begin{figure}
    \centering
    \begin{tabular}{c|c}
        \includegraphics[width=40mm,trim={24cm 0 0 0},clip=True]{Figures/original_dataset.jpg}
        &
        \includegraphics[width=40mm,trim={24.20cm 0 0 0.2cm},clip=True]{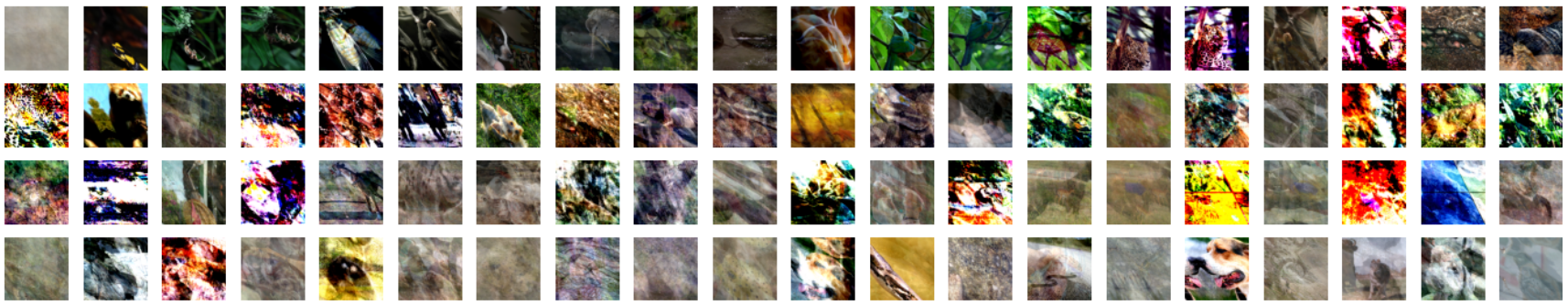}
    \end{tabular}
    \caption{\textit{Left:} Raw input images. \textit{Right:} Reconstruction result with shearing.}
    \label{fig:model_shear}
\end{figure}

\noindent \textbf{\underline{\textit{Shearing.}}} Figure \ref{fig:model_shear} presents the reconstruction from the RTF attack with shearing being used as augmentation for OASIS. We can see that the original image and its sheared version overlap one another in the reconstruction, thereby hindering the attacker from making out the original. This also explains the low PSNR of shearing in Figure \ref{fig:rtf-psnr}.

\noindent \textbf{\underline{\textit{Flipping.}}} Figures \ref{fig:flip_hori} and \ref{fig:flip_vert} illustrate the reconstruction from the RTF attack with horizontal flipping and vertical flipping being used as augmentation for OASIS, respectively. We can see that they did not defend as well against the attack compared to rotation and shearing. A linear combination of an original image and its horizontally or vertically flipped version only generates a reflection of the original, thus the original image is still revealed in the reconstruction. 
Figures \ref{fig:flip_hori} and \ref{fig:flip_vert} show that some images are reflected in the reconstruction. This means that flipping, when used alone, is not the best suited transformation to defend against this class of attacks. However, using flipping in combination with a strong transformation such as rotation or shearing may yield better results.

\begin{figure}
    \centering
    \begin{tabular}{c|c}
        \includegraphics[width=40mm,trim={24cm 0 0 0},clip=True]{Figures/original_dataset.jpg}
        &
        \includegraphics[width=40mm,trim={24.20cm 0 0 0.2cm},clip=True]{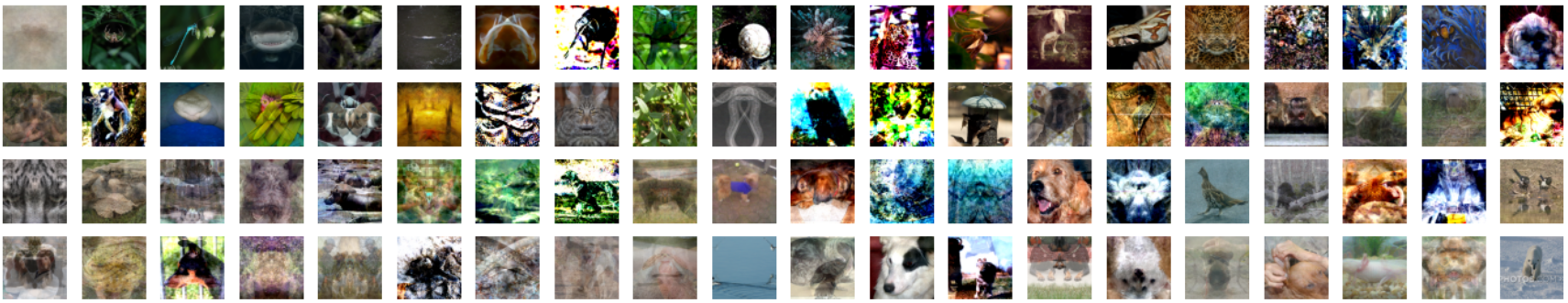}
    \end{tabular}
    \caption{\textit{Left:} Raw input images. \textit{Right:} Reconstruction result with horizontal flipping.}
    \label{fig:flip_hori}
\end{figure}

\noindent \textbf{\underline{\textit{Integrating Major Rotation and Shearing.}}} As previously discussed in Figure \ref{fig:cah-psnr}, an integration of multiple transformations is needed to counter the CAH attack.  Figure \ref{fig:combo} illustrates the reconstruction from CAH when both major rotation and shearing are used in OASIS. It can be seen that all the reconstructed images are unrecognizable and it is impossible to identify any original image from them. This behavior is consistent with the results in Figure \ref{fig:cah-psnr}.

\begin{figure}
    \centering
    \begin{tabular}{c|c}
        \includegraphics[width=40mm,trim={24cm 0 0 0},clip=True]{Figures/original_dataset.jpg}
        &
        \includegraphics[width=40mm,trim={24.20cm 0 0 0.2cm},clip=True]{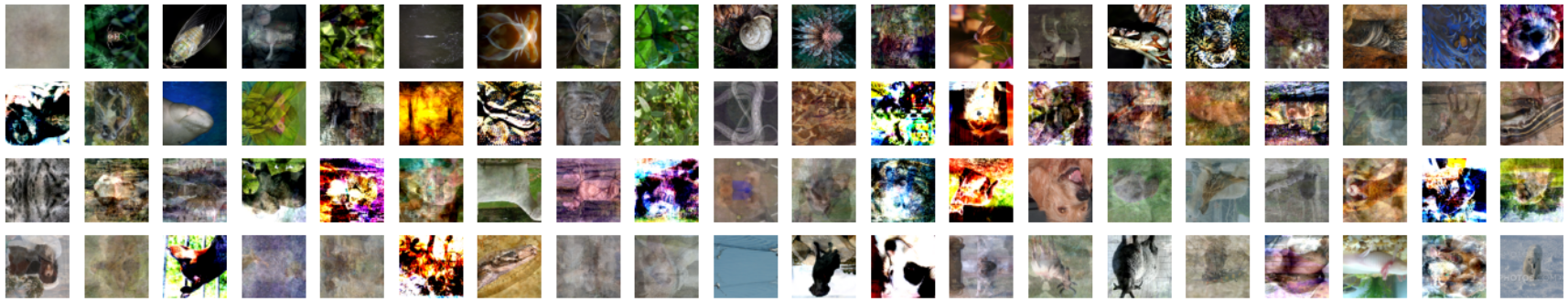}
    \end{tabular}
    \caption{\textit{Left:} Raw input images. \textit{Right:} Reconstruction result with vertical flipping.}
    \label{fig:flip_vert}
\end{figure}

\begin{figure}
    \centering
    \begin{tabular}{c|c}
        \includegraphics[width=37.2mm, angle=90]{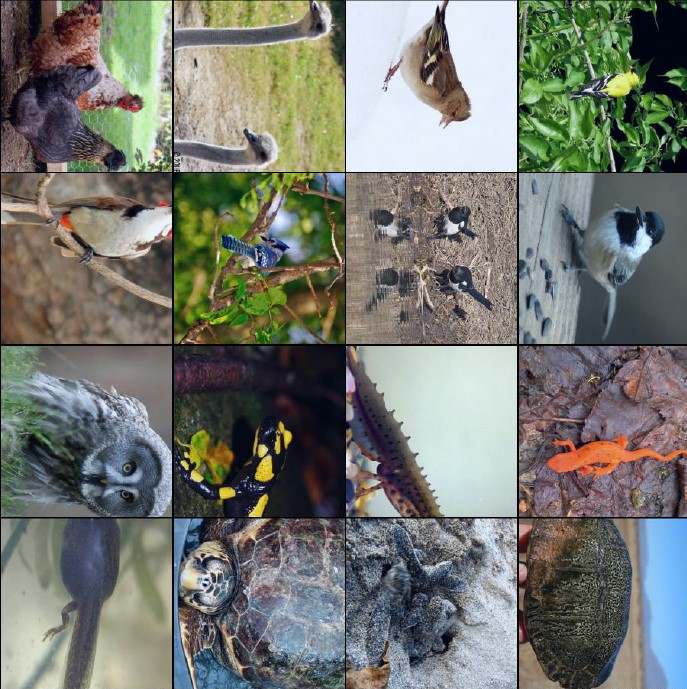}
        &
        \includegraphics[width=37.2mm]{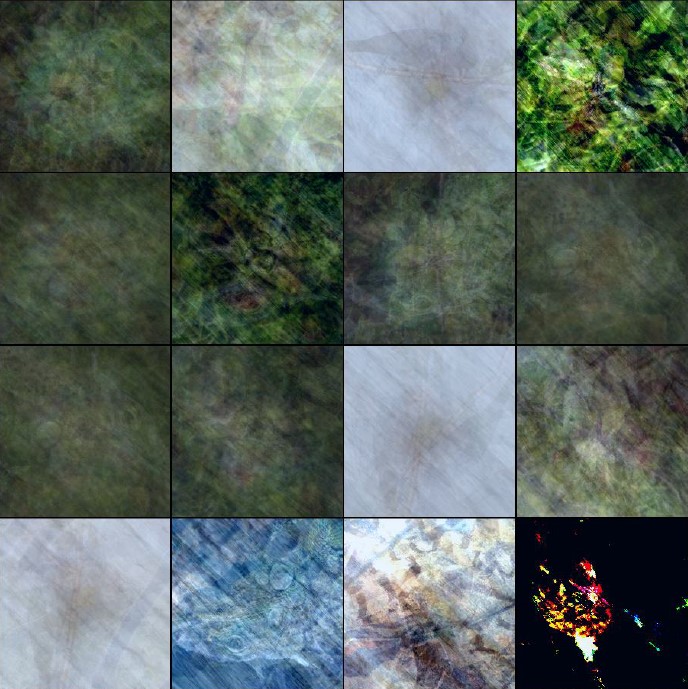}
    \end{tabular}
    \caption{\textit{Left:} Raw input images. \textit{Right:} Reconstruction result with an integration of major rotation and shearing.}
    \label{fig:combo}
\end{figure}



In summary, major rotation and an integration of major rotation and shearing result in the strongest defense against the RTF and CAH attacks, respectively. Additionally, OASIS has been shown to be scalable as it maintains low PSNR on reconstructed images for both small and large batch sizes. We further note that it is not trivial to extract the original image from such an overlap of multiple transformed images without any prior knowledge about certain characteristics of the original image. Although the server might know about certain augmentations being used as a defense, it 
does not know the specific parameters of the transformations (e.g., shearing intensity). Previous research has shown that, even with a mild blurry image, it is 
very challenging to practically reconstruct the original image without knowing the blurring kernel and padding \cite{ren2020neural}, while our defense uses far more complicated and multiple transformations.
\subsection{\textbf{Gradient Inversion Attack on Linear Models.}}
In addition to the RTF \cite{fowl2021robbing} and CAH \cite{boenisch2023curious} attacks, we evaluate our OASIS defense against a reconstruction attack on linear models that was discussed in \cite{fowl2021robbing,geiping2020inverting}. The attack assumes a very restrictive setting where the model is a single-layer and is trained with a logistic regression loss function. Furthermore, the images in each training batch $\mathcal{D}$ are assumed to have unique labels. As users upload their local model updates, the server simply inverts the gradient of each neuron to reconstruct the training images.

Figure \ref{fig:ig-psnr} illustrates the effectiveness of our OASIS defense in reducing the reconstruction quality of this attack. 
Since this is a single-layer model, adding transformed images to the training batch guarantees that $x_t$ and $X'_t$ activate the same neuron, for all $x_t\in\mathcal{D}$. Hence, each reconstructed image will be a linear combination of $x_t$ and $X'_t$. Moreover, such a linear combination hides the content of the original image (as discussed in Section \ref{sec:exp}). Therefore, Figure \ref{fig:ig-psnr} shows that all five transformations yield reconstruction with low PSNR for both datasets and both batch sizes. We can also see that rotation and shearing have better defensive performance than flipping, 
corroborating our findings in Section \ref{sec:exp}.

\begin{figure*}[htp!]
    \centering
    \begin{subfigure}{0.49\linewidth}
        \centering
        \includegraphics[width=0.49\textwidth]{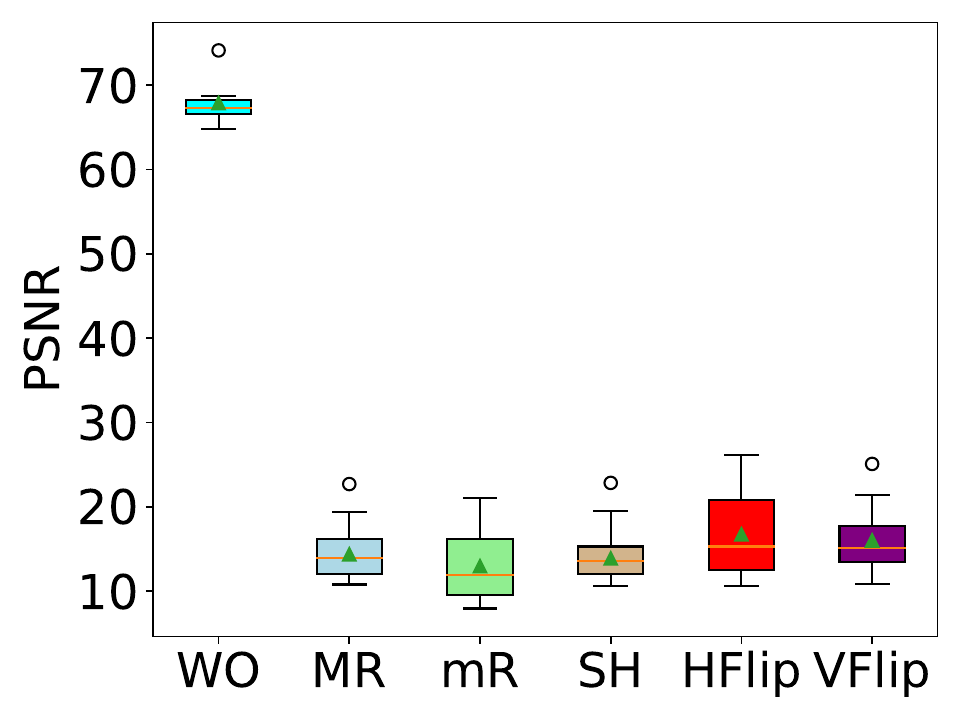}
        \includegraphics[width=0.49\textwidth]{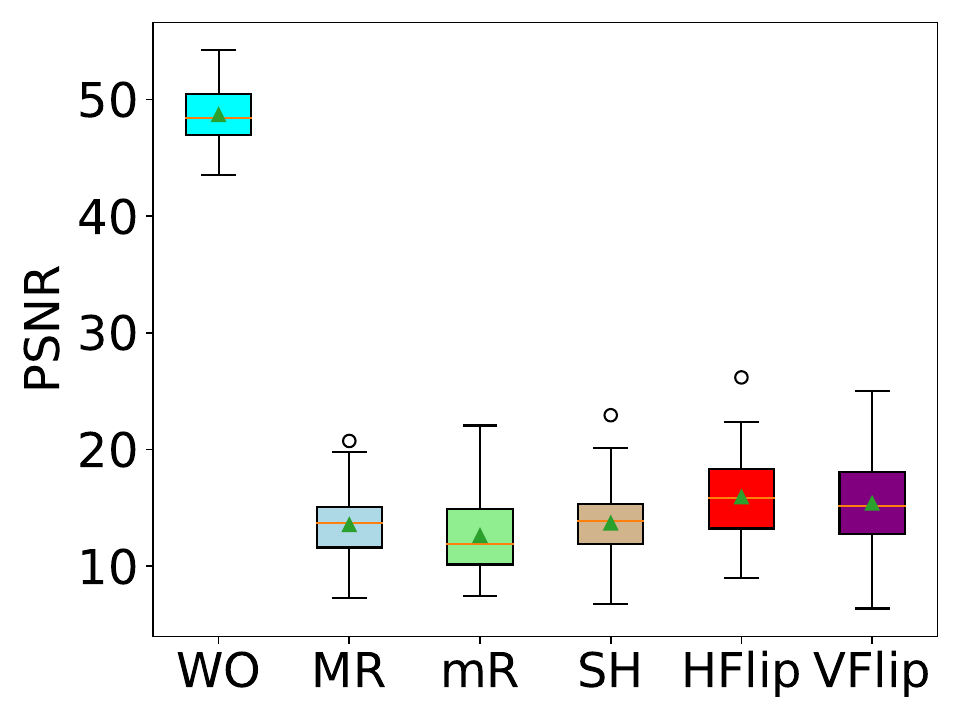}
        \caption{ImageNet. \textbf{Left:} $B = 8$. \textbf{Right:} $B = 64$}
    \end{subfigure}
    \hfill
    \begin{subfigure}{0.49\linewidth}
        \centering
        \includegraphics[width=0.49\textwidth]{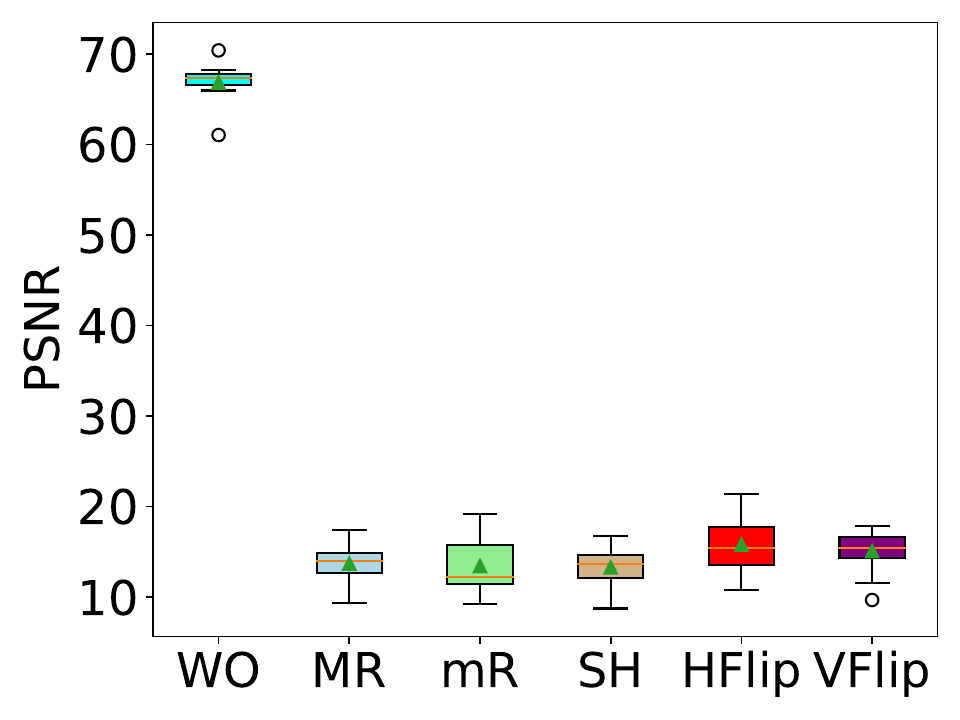}
        \includegraphics[width=0.49\textwidth]{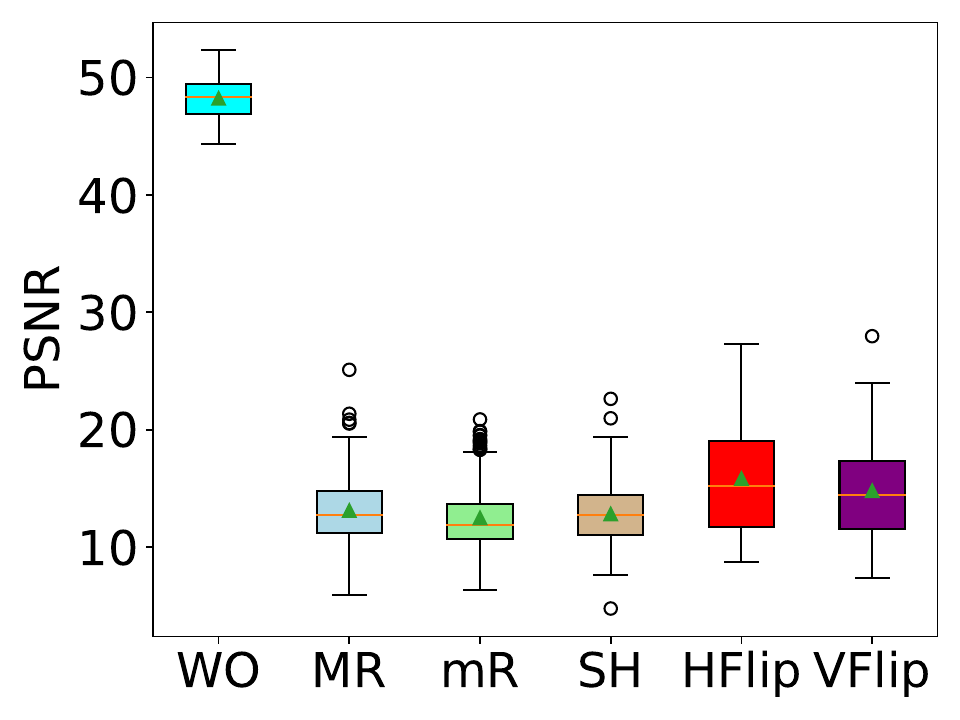}
        \caption{CIFAR100. \textbf{Left:} $B = 8$. \textbf{Right:} $B = 64$}
    \end{subfigure}
    \caption{PSNR values of images reconstructed by the gradient inversion attack on linear models w.r.t different transformations and different batch sizes on ImageNet and CIFAR100. The green triangle denotes the average PSNR over all reconstructed images. \emph{(WO = Without OASIS, MR = Major Rotation, mR = Minor Rotation, SH = Shearing, HFlip = Horizontal Flip, and VFlip = Vertical Flip)} 
    }
    \label{fig:ig-psnr}
\end{figure*}

\subsection{\textbf{Impact of OASIS on Model Performance}}
We 
measure the effect of using OASIS on the training models as it alters the input dataset for training. For this experiment, we train ResNet-18 models \cite{he2016deep} on 
ImageNet and CIFAR100, then compare the final testing accuracies when training with and without OASIS. In particular, when training with OASIS, we replace each training batch $\mathcal{D}$ with $\mathcal{D}'$ as mentioned in Section \ref{sec:oasis}. 
The result for each transformation is shown in Table \ref{table:mp-table}.


For ImageNet \cite{deng2009imagenet}, we extract a subset of 10 classes: tench, English springer, cassette player, chain saw, church, French horn, garbage truck, gas pump, golf ball, and parachute\footnote{\url{https://github.com/fastai/imagenette}}. Then, we evaluate the model performance on classifying those 10 classes. 
Using our ResNet-18 architecture, 
we train for 100 epochs with an Adam optimizer at a learning rate of 0.001 and weight decay of $10^{-5}$.

With regard to CIFAR100 \cite{krizhevsky2009learning}, we use its original classification task with 100 classes. 
Again, using our ResNet-18 architecture, 
we train for 120 epochs with an Adam optimizer at a learning rate of 0.001 and weight decay of $10^{-2}$.



\begin{table}[htp!]
    \centering
    \caption{Comparing model accuracy (\%) when training with and without OASIS}
    \label{table:mp-table}
  \begin{tabular}{lSSSSSS}
    \toprule
    \multirow{2}{*}{Transformation} &
      \multicolumn{2}{c}{Dataset} \\
      & {ImageNet} & {CIFAR100} \\
      \midrule
    Major Rotation & 92.6 & 74.3\\
    Minor Rotation & 92.6 & 74.1\\
    Shearing & 95.4 & 73.7\\
    Horizontal Flip & 94.0 & 75.1\\
    Vertical Flip & 94.8 & 74.3\\
    Major Rotation + Shearing & 90.9 & 74.6\\
    \midrule
    Without OASIS & 94.8 & 75.2\\
    \bottomrule
  \end{tabular}
\end{table}
 
Across all the transformations, OASIS does not impose any major degradation on the model accuracy. The accuracy is still maintained over 90\% on ImageNet, and drops \textit{at most} 1.5\% on CIFAR100. The reason for this is that image augmentation methods are originally developed for improving the generalization and reducing overfitness of ML models. From this, the claims made in Section \ref{sec:settings} are confirmed. 
\section{Related Work}\label{sec:related}
\noindent \textbf{Data Reconstruction Attacks.} Reconstruction attacks have been one of the main topics of interest in ML security and privacy. Over the decade, various kinds of reconstruction attacks have been proposed, including class-wise representation-based attacks \cite{wang2019beyond,hitaj2017deep,sun2021soteria} and optimization-based attacks \cite{zhu2019deep,zhao2020idlg,yin2021see}. However, in the context of FL, most of these attacks are not able to exploit the full capability of dishonest servers. Recent work 
\cite{fowl2021robbing,boenisch2023curious} devises a new class of \textit{active} reconstruction attacks that has been shown to significantly outperform prior attacks by having the dishonest server manipulate the global model parameters to its advantage. For that reason, this new class of attack remains a critical and practical threat for FL. 
Our work focuses on devising a general defense that effectively protects user data against these attacks. From analyzing the underlying principle of gradient inversion, our defense OASIS is designed to minimize 
reconstruction quality.

\noindent \textbf{Current Defenses.} Presently, there is no existing defense that can defend against \textit{active} reconstruction attacks via \textit{dishonest servers}. 
In general, previous defenses 
utilize a threat model with an honest-but-curious server that is substantially weaker than our threat model which includes an actively dishonest server. Several defense mechanisms have been proposed to tackle data reconstruction attacks in general, but they remain ineffective in countering the versions presented in this paper. 
Through gradient compression and sparsification methods, the work in \cite{zhu2019deep,sun2021soteria} \emph{pruned} gradients with negligible magnitudes to zero. Nonetheless, even in a case where the 
majority of the gradients are pruned, data extracted is 
still recognizable \cite{boenisch2023curious}. 

\begin{figure}
    \centering
    \begin{tabular}{c|c}
        \includegraphics[width=40mm]{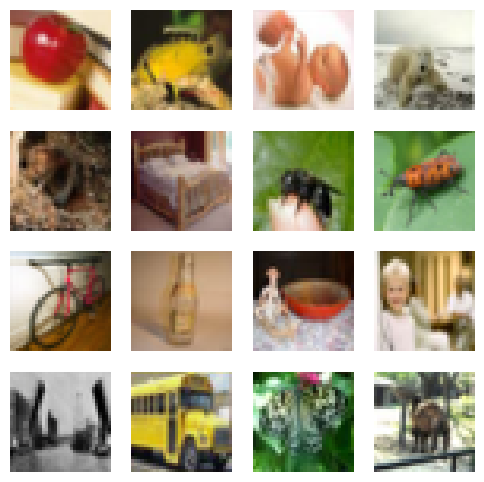}
        &
        \includegraphics[width=40mm]{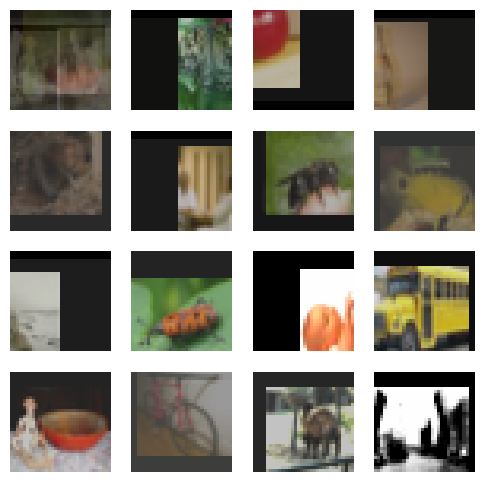}
    \end{tabular}
    \caption{Reconstruction result of RTF against the defense in \cite{gao2021privacy}. The content of the original images is revealed in the reconstruction. \textit{Left:} Raw input images. \textit{Right:} Reconstruction.}
    \label{fig:ats}
\end{figure}

Gao et al. \cite{gao2021privacy} leverage image augmentation in their proposed defense, but it
can only tackle optimization-based attacks. In particular, the defense replaces each image in the dataset with a transformed image so that the objective function of the attacks becomes more difficult to solve. However, it fails to counter the \textit{active} reconstruction attacks since their principle (Section \ref{sec:threat}) still applies: if an attacked neuron is activated only by one transformed image, the image would be reconstructed. 
To support this claim, we conduct an experiment in which we launch the RTF attack \cite{fowl2021robbing} against this defense and illustrate the resulting reconstruction in Figure \ref{fig:ats} (we adopt the implementation of \cite{gao2021privacy} from \url{https://github.com/gaow0007/ATSPrivacy}). 
As can be seen, the reconstruction reveals the content of the original input images. Therefore, defenses against optimization-based reconstruction attacks are not robust against these \textit{active} reconstruction attacks if they do not address the attack principle of gradient inversion. 

In \cite{fowl2021robbing,boenisch2023curious}, the authors evaluate the use of DP as a defense, and show that it imposes a major degradation on the model accuracy and the reconstructed images are still recognizable. Our OASIS defense is proven to effectively counter this new class of attacks as it tackles the core attack principle. Moreover, OASIS imposes minimal impact on model performance.
\section{Conclusion}\label{sec:con}
In this paper, we have revealed the key principle behind active reconstruction attacks in Federated Learning (FL) and have theoretically shown how to tackle this class of attacks. With machine learning foundations in data preprocessing, we have proposed OASIS, a novel method to augment images in a way such that an \textit{actively} dishonest server is unable to memorize individual gradient parameters, but a linear combination of an image and its augmented counterparts. In doing so, we \emph{offset} the active reconstruction attacks, rendering reconstructions unrecognizable. To address FL's promise of maintaining model performance, we also demonstrate that the expansion of a labeled dataset through augmentation preserves and, in some cases, improves model performance. From our evaluation, OASIS stands as a general, viable, and scalable solution to truly promote and reinforce the guarantees of FL. Although the use of image augmentation makes OASIS confined to the image domain, we note that the attack principle that we uncover in Section \ref{sec:threat} is not limited to any data types. 
Future work will focus on finding alternative methods besides image augmentation to implement an effective defense for tabular and textual data.

\section*{Acknowledgement}

This work was supported in part by the National Science Foundation under grants CNS-2140477, CNS-2140411, and ITE-2235678, and by the Korea Institute of Energy Technology Evaluation and Planning (KETEP) grant funded by the Korea government (MOTIE) (RS-2023-00303559, Study on developing cyber-physical attack response system and security management system to maximize real-time distributed resource availability).

This work was authored in part by the National Renewable Energy Laboratory, operated by Alliance for Sustainable Energy, LLC, for the U.S. Department of Energy (DOE) under Contract No. DE-AC36-08GO28308. The views expressed in the article do not necessarily represent the views of the DOE or the U.S. Government. The U.S. Government retains and the publisher, by accepting the article for publication, acknowledges that the U.S. Government retains a nonexclusive, paid-up, irrevocable, worldwide license to publish or reproduce the published form of this work, or allow others to do so, for U.S. Government purposes.

\bibliographystyle{ieeetr}
\bibliography{ref}

\end{document}